\documentclass[a4paper]{article}

\usepackage[utf8]{inputenc}
\usepackage[francais,american]{babel}
\usepackage{enumerate}
\usepackage{amsmath,amstext,amsfonts,amssymb,amsthm,xfrac,paralist}
\usepackage{url}
\usepackage{xspace}

\usepackage{tikz}
\usepackage{subfig}
\usepackage{array}
\usepackage{verbatim}
\usepackage{times}
\usepackage[ruled,vlined,boxed,commentsnumbered]{algorithm2e}

 \usepackage{tikz}

 \usepackage{fancybox}
\usepackage{fp}

\usepackage{graphicx}

\tikzstyle{fancytitle} = [fill=blue!40, text=black, rounded corners,inner sep=4pt] 
\tikzstyle{mybox} = [draw=blue!40, fill=blue!20, very thick, rectangle, rounded corners, inner ysep=10pt, drop shadow]
\tikzstyle{fancytitleR} = [fill=bluegreen, text=black, rounded corners,inner sep=4pt] 
\tikzstyle{myboxTitle} = [draw=bluetheme, fill=bluetheme, very thick, rectangle, rounded corners, inner ysep=10pt, drop shadow]
\tikzstyle{myboxR} = [draw=bluegreen, fill=bluegreen!20, very thick, rectangle, rounded corners, inner ysep=10pt]
 \usepackage{calc}

\usetikzlibrary{positioning}
\usetikzlibrary{calc}
\usetikzlibrary{backgrounds}

\newcommand{\pr}[1]{\ensuremath{p_{#1}}}

\newcommand{\inSched}{
\setcounter{nbProcs}{0}
\FPadd{\cury}{\procspace}{1}
}

\newcommand{\addProc}{
\addtocounter{nbProcs}{1}
\FPsub{\cury}{\cury}{\procspace}
\FPsub{\cury}{\cury}{1}
\draw[draw=none] (-1.5,\cury) rectangle (-0.5,\cury-1) node[midway] {\pr{\thenbProcs}};
\FPeval{\curx}{0}
}

\newcommand{\addProcChain}{
\addtocounter{nbProcs}{1}
\FPsub{\cury}{\cury}{\procspace}
\FPsub{\cury}{\cury}{1}
\FPsub{\cury}{\cury}{\procspace}

\draw[draw=none] (-1.5,\cury) rectangle (-0.5,\cury-1) node[midway] {$\pr{1}^{eq}$};
\FPeval{\curx}{0}
}

\newcommand{\addJobSpeed}[4][anchor=center]{ % First instruction is the size (weight), second instruction is the execution speed, third instruction is the name.  The optionnal argument is the position

\draw[fill=white] (\curx,\cury-1) rectangle ($(\curx+#2/#3,\cury+#3-1)$) node[#1,midway,very thick] {$w_#4$}; 
\FPdiv{\curz}{#2}{#3}
\FPadd{\curx}{\curx}{\curz}
}

\newcommand{\addLisJob}[4][midway]{ % First instruction is the size (weight), second instruction is the execution speed, third instruction is the name.
\draw (\curx,\cury-1) rectangle ($(\curx+#2/#3,\cury+#3-1)$) node[#1] {$w_#4$}; 
\FPsub{\cury}{\cury}{\procspace}
\FPsub{\cury}{\cury}{1}
}

\theoremstyle{plain}

\newtheorem{lemma}{Lemma}
\newtheorem{theorem}{Theorem}
\newtheorem{proposition}{Proposition}
\newtheorem{corollary}{Corollary}

\theoremstyle{definition}

\theoremstyle{remark}

\newtheorem*{remark}{Remark}

\newcommand{\fmax}{\ensuremath{f_{\max}}\xspace}
\newcommand{\fmin}{\ensuremath{f_{\min}}\xspace}
\newcommand{\finf}{\ensuremath{f_{\inf,i}}\xspace}
\newcommand{\fr}{\ensuremath{f_{\texttt{rel}}}\xspace}
\newcommand{\freex}{\ensuremath{f_{\texttt{re-ex}}}\xspace}

\newcommand{\energy}{\textsc{Energy}} %\xspace}
\newcommand{\exe}{\ensuremath{\mathcal{E}\!xe}\xspace}
\newcommand{\dff}{\textsc{Dec\-reasing-First-Fit}\xspace}
\newcommand{\partition}{2-\textsc{Partition}\xspace}
\newcommand\II{\ensuremath{\mathcal{I}}\xspace}
\newcommand\EE{\ensuremath{\mathcal{E}}\xspace}

\newcommand{\approxchain}{\textsc{Ap\-prox-Chain}}
\newcommand{\xopt}{\textsc{X-Opt}}
\newcommand{\merge}{\textsc{Merge-Lists}}
\newcommand{\add}{\textsc{Add-List}}
\newcommand{\trim}{\textsc{Trim}}
\newcommand{\computeVl}{\textsc{Compute\_}$V_l$}

\newcommand{\tricrit}{\textsc{Tri-Crit}\xspace}
\newcommand{\chain}{\textsc{Tri-Crit-Chain}\xspace}
\newcommand{\indep}{\textsc{Tri-Crit-In\-dep}\xspace}

\let\ab\allowbreak

\title{Approximation algorithms for energy, reliability and makespan
  optimization problems}

\author{Guillaume Aupy\\
   LIP, ENS Lyon, France\\
   Guillaume.Aupy@ens-lyon.fr \\
   \and
   Anne Benoit \\
   LIP, ENS Lyon, France \& Institut Universitaire de France, Paris, France \\
   Anne.Benoit@ens-lyon.fr \\
   }

\date{\today}

\begin{document}
\maketitle

\begin{abstract}
  In this paper, we consider the problem of scheduling an application
  on a parallel computational platform. The application is a
  particular task graph, either a linear chain of tasks, or a set of
  independent tasks. The platform is made of identical processors,
  whose speed can be dynamically modified. It is also subject to
  failures: if a processor is slowed down to decrease the energy
  consumption, it has a higher chance to fail. Therefore, the
  scheduling problem requires to re-execute or replicate tasks (i.e.,
  execute twice a same task, either on the same processor, or on two
  distinct processors), in order to increase the reliability. It is a
  tri-criteria problem: the goal is to minimize the energy
  consumption, while enforcing a bound on the total execution time
  (the makespan), and a constraint on the reliability of each task. 

  Our main contribution is to propose approximation algorithms for
  these particular classes of task graphs.  For linear chains, we
  design a fully polynomial time approximation scheme. However, we
  show that there exists no constant factor approximation algorithm
  for independent tasks, unless P=NP, and we are able in this case to
  propose an approximation algorithm with a relaxation on the makespan
  constraint.
%  \keywords{Scheduling \and energy \and reliability \and makespan \and models \and approximation algorithms }
\end{abstract}

\section{Introduction}
\label{sec.intro}

Energy-awareness is now recognized as a first-class constraint in the
design of new scheduling algorithms. To help reduce energy
dissipation, current processors from AMD, Intel and Transmetta allow
the speed to be set dynamically, using a dynamic voltage and frequency
scaling technique (DVFS). Indeed, a processor running at speed $s$
dissipates $s^3$ watts per unit of time \cite{pow3IPDPS}.  However,
it has been recognized that reducing the speed of a processor has a
negative effect on the reliability of a schedule: if a processor is
slowed down, it has a higher chance to be subject to transient
failures, caused for instance by software errors 
\cite{Zhu04EEM,Degal05SEI}. 

Motivated by the application of speed scaling on large scale machines
\cite{Oliner04}, we consider a tri-criteria problem
energy/reliability/makespan: the goal is to minimize the energy
consumption, while enforcing a bound on the makespan, i.e., the total
execution time, and a constraint on the reliability of each task.  The
application is a particular task graph, either a linear chain of
tasks, or a set of independent tasks. The platform is made of
identical processors, whose speed can be dynamically modified. 

In order to make up for the loss in reliability due to the energy
efficiency, we consider two standard techniques: \emph{re-execution}
consists in re-executing a task twice on a same processor
\cite{Zhu04EEM,Zhu06}, while \emph{replication} consists in executing
a same task on two distinct processors simultaneously
\cite{Assayad11}. We do not consider \emph{checkpointing}, which
consists in ``saving'' the work done at some points, hence reducing
the amount of work lost when a failure occurs
\cite{Melhem03CP,Zhang03CP}.
 
The schedule therefore requires to (i) decide which tasks are
re-executed or replicated; (ii) decide on which processor(s) each task
is executed; (iii) decide at which speed each processor is processing
each task. For a given schedule, we can compute the total execution
time, also called {\em makespan}, and it should not exceed a
prescribed deadline. Each task has a reliability that can be computed
given its execution speed and its eventual replication or
re-execution, and we must enforce that the execution of each task is
reliable enough.  Finally, we aim at minimizing the energy
consumption.  Note that we consider a set of homogeneous processors,
but each processor may run at a different speed; this corresponds to
typical current platforms with DVFS.

\paragraph{Related work.}
The problem of minimizing the energy consumption without exceeding a
given deadline, using DVFS, has been widely studied, without
accounting for reliability issues.  The problem for a linear chain of
tasks is known to be solvable in polynomial time in this case, see
\cite{aupy12ccpe}. \cite{Alon97} showed that the problem of
scheduling independent tasks can be approximated by a factor
$(1+\varepsilon)$: they exhibit a polynomial time approximation scheme
(PTAS). \cite{RenaudGoudGreedy} studied the performance of greedy
algorithms for the problem of scheduling independent tasks, with the
objective of minimizing the energy consumption, and proposed some
approximation algorithms.

All these work do not account for reliability issues. However,
\cite{Zhu04EEM} showed that reducing the speed of a processor
increases the number of transient failure rates of the system; the
probability of failures increases exponentially, and this probability
cannot be neglected in large-scale computing \cite{Oliner04}.  Few
authors have tackled the tri-criteria problem including reliability,
and to the best of our knowledge, there are no approximation
algorithms for this problem. \cite{Zhu06} initiated the study of this
problem, using re-execution. However, they restrict their study to the
scheduling problem on a single processor, and do not try to find any
approximation ratio on their algorithm. \cite{Assayad11} have
recently proposed an off-line tri-criteria scheduling heuristic (TSH),
which uses replication to minimize the makespan, with a
threshold on the global failure rate and the maximum power
consumption.  TSH is an improved critical-path list sche\-duling
heuristic that takes into account power and reliability before
deciding which task to assign and to replicate onto the next free
processors. However, the complexity of this heuristic is unfortunately
exponential in the number of processors, and the authors did not try
to give an approximation ratio on their heuristic.  Finally,
\cite{rr7757} also study the tri-criteria problem, but from an
heuristic point of view, without trying to ensure any approximation
ratio on their heuristics. Moreover, they do not consider replication
of tasks, but only re-execution as in \cite{Zhu06}.  However, they
present a formal model of the tri-criteria problem, re-used in this
paper.

Finally, there is some related work specific to the problem of
independent tasks, since several approximation algorithms have been
proposed for variants of the problem. One may try to minimize the
$\ell_k$ norm, i.e., the quantity $(\sum_{q=1}^p (\sum_{i\in load(q)}
a_i)^k)^{1/k}$, with $p$ processors, where $i\in load(q)$ means that
task~$T_i$ is assigned to processor~$q$, and $a_i$ is the weight of
task~$T_i$ \cite{Alon97}.  Minimizing the power consumption
then amounts to minimize the $\ell_3$ norm
\cite{RenaudGoudGreedy}, and the problem of makespan minimization is
equivalent to minimizing  
%minimizing the makespan amounts to minimize 
the $\ell_{\infty}$ norm, i.e., minimize $\max_{1\leq q \leq p}
\sum_{i\in load(q)} a_i$ \cite{Graham69,Ausiello99}.  These problems
are typical {\em load balancing} problems, in which the load
(computation requirement of the tasks) must be balanced between
processors, according to various criteria.

%When considering scheduling of independent tasks, it is very natural to think 
%of minimizing the $\ell_k$ norm for some $k$. For example, scheduling 
%independent tasks on $p$ processors in order to minimize 
%\begin{compactitem}
%    \item the makespan: minimizing the $\ell_{\infty}$ norm (see \cite{Graham69,Ausiello99});
%    \item the energy (without reliability): minimizing the $\ell_2$ norm 
%(see \cite{RenaudGoudGreedy});
%    \item or even any $\ell_k$ norm (see \cite{Alon97}).
%\end{compactitem}
%Those cases have been widely studied, with good approximation factors (these
%problems all belong to the class PTAS).

\paragraph{Main contributions.}  
In this paper, we investigate the tri-criteria problem of minimizing
the energy with a bound on the makespan and a constraint on the
reliability.  First in Section~\ref{sec.fw}, we formally introduce
this tri-criteria scheduling problem, based on the previous models
proposed by \cite{Zhu06} and~\cite{rr7757}.  To the best of our
knowledge, this is the first model including both re-execution and
replication in order to deal with failures. The main contribution of
this paper is then to provide approximation algorithms for some
particular instances of this tri-criteria problem.

For linear chains of tasks, we propose a fully polynomial time
approximation scheme (Section~\ref{sec.lin}). Then in
Section~\ref{sec.indep}, we show that there exists no constant factor
approximation algorithm for the tri-criteria problem with independent
tasks, unless P=NP. We prove that by relaxing the constraint on
the makespan, we are able to give a polynomial time constant factor
approximation algorithm.  To the best of our knowledge, these are the
first approximation algorithms for the tri-criteria problem.

\section{Framework}
    \label{sec.fw}
Consider an application task graph $\mathcal{G}=(V,\mathcal{E})$,
where $V= \{T_1, T_2, \dots, T_n\}$ is the set of tasks, $n = |V|$,
and where $\EE$ is the set of precedence edges between tasks. For $1
\leq i \leq n$, task~$T_i$ has a weight~$w_i$, that corresponds to the
computation requirement of the task.  $S=\sum_{i=1}^n w_i$ is the
sum of the computation requirements of all tasks.

The goal is to map the task graph onto $p$ identical processors, with
the objective of minimizing the total energy consumption, while
enforcing a bound on the total execution time (makespan), and matching
a reliability constraint. Processors can have arbitrary speeds,
determined by their frequency, that can take any value in the interval
$[\fmin,\fmax]$ (dynamic voltage and frequency scaling with continuous
speeds). Higher frequencies, and hence faster speeds, allow for a
faster execution, but they also lead to a much higher (supra-linear)
power consumption.  Moreover, reducing the frequency of a processor
increases the number of transient failures of the system. Therefore,
some tasks are executed once at a speed high enough to satisfy the
reliability constraint, while some other tasks are executed several
times (either on the same processor, or on different processors), at a
lower speed. We detail below the conditions that are enforced on the
corresponding execution speeds. The problem is therefore to decide
which tasks should be executed several times, on which processor, and
at which speed to run each execution of a task, as well as the
schedule, i.e., in which order the tasks are executed on each
processor.  Note that \cite{rr7757} showed that it is always better
to execute a task at a single speed, and therefore we assume in the
following that each execution of a task is done at a single speed.

We now detail the three objective criteria (makespan, reliability,
energy), and then define formally the problem.

\subsection{Makespan}

The makespan of a schedule is its total execution time. The first task
is scheduled at time $0$, so that the make\-span of a schedule is simply
the maximum time at which one of the processors finishes its
computations.  Given a schedule, the makespan should not exceed the
prescribed deadline~$D$. 

Let $\exe(w_i,f)$ be the execution time of a task $T_i$ of
weight~$w_i$ at speed~$f$. We assume that the cache size is adapted to
the application, therefore ensuring that the execution time is
linearly related to the frequency \cite{Melhem03CP}: $\exe(w_i,f) =
\frac{w_i}{f}$. Note that we consider a worst-case scenario, and the
deadline~$D$ must be matched even in the case where all tasks that are
scheduled to be executed several times fail during their first
executions, hence all execution times for a same task should be
accounted for.

\subsection{Reliability}
    \label{sec.rel}
To define the reliability, we use the failure model of
\cite{Zhu04EEM} and \cite{Zhu06}. 
\emph{Transient} failures are failures caused by software errors for example. 
They invalidate only the execution of
the current task and the processor subject to that failure will be
able to recover and execute the subsequent tasks assigned to it (if
any). In addition, we use the reliability model introduced by \cite{Shatz89}, 
which states that the radiation-induced
transient failures follow a Poisson distribution.  The parameter
$\lambda$ of the Poisson distribution is then %:\\ %[-.2cm]
%\begin{equation}
%	\label{fault}
$\lambda(f)=\tilde{\lambda_0} \;  e^{\tilde{d}\frac{\fmax-f}{\fmax-\fmin}}$,
%\end{equation}
where $\fmin\leq f \leq \fmax$ is the processing speed, the exponent
$\tilde{d}\geq0$ is a constant, indicating the sensitivity of failure
rates to dynamic voltage and frequency scaling, and
$\tilde{\lambda_0}$ is the average failure rate at speed~\fmax.  We
see that reducing the speed for energy saving increases the failure
rate exponentially.  The reliability of a task~$T_i$ executed once at
speed $f$ is $$R_i(f)=e^{-\lambda(f)\times\exe(w_i,f)}.$$  Because the
failure rate~$\tilde{\lambda_0}$ is usually very small, of the order of
$10^{-5}$ per time unit \cite{Assayad11}, or even $10^{-6}$
\cite{Baleani03,Izo07}, we can use the first order approximation of
$R_i(f)$ as %[-.2cm]
%\begin{equation}
%	\label{rel_first_order}
\begin{align*}
R_i(f) &= 1-\lambda(f)\times\exe(w_i,f) \\
       &= 1-\tilde{\lambda_0}\;
       e^{\tilde{d}\frac{\fmax-f}{\fmax-\fmin}} \times \frac{w_i}{f}
       \\
       &= 1-\lambda_0\;e^{-df} \times \frac{w_i}{f},
       \end{align*}
%\end{equation}
where $d=\frac{\tilde{d}}{\fmax-\fmin}$ and $\lambda_0 =
\tilde{\lambda_0} e^{d\fmax}$. 

\medskip
Note that this equation holds if \mbox{$\varepsilon_{i} = \lambda(f) \times
  \frac{w_i}{f} \ll 1$}.  With, say, $\lambda(f) = 10^{-5}$, we need
$\frac{w_i}{f} \leq 10^{3}$ to get an accurate approximation with
$\varepsilon_{i} \leq 0.01$: the task should execute within $16$
minutes. In other words, large (computationally demanding) tasks
require reasonably high processing speeds with this model (which makes
full sense in practice).

We want the reliability~$R_i$ of each task $T_{i}$ to be greater than
a given threshold, namely $R_{i}(\fr)$, hence enforcing a local
constraint dependent on the task: $R_i \geq R_{i}(\fr)$.   
If task~$T_{i}$ is executed only once at speed~$f$, then the
reliability of~$T_i$ is $R_i=R_i(f)$. Since the reliability increases
with speed, we must have $f\geq \fr$ to match the reliability
constraint.
If task~$T_{i}$ is executed twice (speeds~$f^{(1)}$ and~$f^{(2)}$),
then the execution of~$T_{i}$ is successful if and only if one of the
attempts do not fail, so that the reliability of~$T_{i}$ is $R_{i}= 1
- (1 - R_i(f^{(1)}))( 1 - R_i(f^{(2)}))$, and this quantity should be
at least equal to $R_{i}(\fr)$. 

We restrict in this work to a maximum of two executions of a same task,
either on the same processor (what we call {\em re-execution}), or on
two distinct processors (what we call {\em replication}).  This is
based on the following observation on the two cases in which a third
execution of a task may be useful. 
\begin{compactenum}
\item The deadline is such that even if all tasks are executed
  twice at the slowest possible speed, the execution time is still
  lower than the deadline. Then, the problem is to decide which task
  should be executed three times, and it is quite similar to the
  problem that we discuss in this paper. 
%long enough such that all tasks are executed
%  twice and some time remains, but in that case one could do an
%  equivalent case study than this one since the problem should be to
%  decide which tasks should be executed three times;
\item Some tasks are too big to be re-executed while there remains
  some time such that some small tasks can be executed at least three
  times at a speed even slower. In this case, the gain in energy
  consumption is negligible compared to the energy consumption of the
  big tasks at speed~\fr.
\end{compactenum}
%Because of those two arguments, we restricted our model to a maximum
%of two execution per tasks. %Needless to say that it also simplifies
%the model that is already 

%\todo{+Introduce \finf, the minimal speed such that a task executed two times 
%matches the reliability constraints.}
Note that if both execution speeds are equal, i.e.,
$f^{(1)}=f^{(2)}=f$, then the reliability constraint writes
$1-(\lambda_0 w_i\frac{ e^{-df}}{f})^2 \geq R_{i}(\fr)$, and
therefore $$\lambda_0 w_i \frac{ e^{-2df}}{f^2} \leq \frac{
  e^{-d\fr}}{\fr}\; .$$ In the following, $\finf$ is the solution to
the equation $\lambda_0 w_i \frac{e^{-2d\finf}}{(\finf)^{2}} =
\frac{e^{-d\fr}}{\fr}$, and hence task~$T_i$ can be executed twice at
a speed greater than or equal to $\finf$ while meeting the reliability
constraint. In practice, $\finf$ is small enough so that tasks are
usually executed faster than this speed, hence reinforcing the
argument that it is meaningful to restrict to two executions of a same
task.

\subsection{Energy}

The total energy consumption corresponds to the sum of the energy
consumption of each task. Let $E_i$ be the energy consumed by
task $T_i$. For one execution of $T_i$ at speed~$f$, the
corresponding energy consumption is $E_i(f) =\exe(w_i,f) \times f^3 =
w_i \times f^2$, which corresponds to the dynamic part of the
classical energy models of the literature \cite{pow3IPDPS,BKP07}.  
Note that we do not take static energy into account, because all 
processors are up and alive during the whole execution.

If task $T_{i}$ is executed only once at speed~$f$, then $E_{i} =
E_i(f)$.  Otherwise, if task $T_{i}$ is executed twice at speeds
$f^{(1)}$ and $f^{(2)}$, it is natural to add up the energy consumed
during both executions, just as we consider both execution times when
enforcing the deadline on the makespan. Again, this corresponds to the
worst-case execution scenario.  We obtain $E_i = E_i(f^{(1)}_i) +
E_i(f^{(2)}_i)$.  Note that some authors \cite{Zhu06} consider only
the energy spent for the first execution in the case of re-execution,
which seems unfair: re-execution comes at a price both in the makespan
and in the energy consumption.  Finally, the total energy consumed by
the schedule, which we aim at minimizing, is $E = \sum_{i=1}^{n} E_i$.

% The total energy consumption corresponds to the sum of the energy
% consumption of each task. Let $E_i$ be the energy consumed by
% task~$T_i$. For one execution of task~$T_i$ at speed~$f$, the
% corresponding energy consumption is $E_i(f) =\exe(w_i,f) \times f^3 =
% w_i \times f^2$, 
% which corresponds to the dynamic part of the classical energy  
% models of the literature~\cite{pow3IPDPS,BKP07,pow3ICPP,aupy12ccpe}.
% Note that we do not take static energy into account, because
% all processors are up and alive during the whole execution.

% If task~$T_{i}$ is executed only once at speed~$f$, then 
% $E_{i} = E_i(f)$. 
% Otherwise, if task~$T_{i}$ is executed twice at speeds $f^{(1)}$ and $f^{(2)}$, 
% it is natural to add up the energy consumed during both executions because
% executions are independent from one another.
% When two executions for a same task are on the same processor, this 
% corresponds to the worst-case execution scenario.  We obtain 
% $E_i = E_i(f^{(1)}_i) + E_i(f^{(2)}_i)$.

% The total energy consumed by the schedule is
% \begin{equation}
% \label{eq.energy}
% 	\energy: ~~~ E = \sum_{i=1}^{n} E_i.
% \end{equation}

    \subsection{Optimization problem}
\label{opt_problem}
%The natural optimization problem is 
%\begin{definition}[
%\tricrit: 
Given an application graph $\mathcal{G}=(V,\mathcal{E})$ and $p$
identical processors, \tricrit is the problem of finding a schedule
that specifies which tasks should be executed twice, on which
processor and at which speed each execution of a task should be
processed, such that the total energy consumption~$E$ is minimized,
subject to the deadline~$D$ on the makespan and to the local reliability
constraints $R_i \geq R_i(\fr)$ for each $T_i \in V$.
%  minimizing the energy 
%   consumption of Equation~(\ref{eq.energy}), subject to the deadline bound~$D$ 
%   and to the reliability constraint, while using replication and/or 
%   re-execution. 
%\end{definition}

We focus in this paper on the two following sub-problems that are 
restrictions of \tricrit to special application graphs:
%\begin{definition}[
\begin{itemize}
\item \chain: the graph is such that 
\\$\mathcal{E} = \cup_{i=1}^{n-1} \{ T_i \rightarrow T_{i+1}\}$; 

\item \indep: 
  the graph is %defined by $V= \{T_1,\dots,T_n\}$ and 
such that $\mathcal{E}=\emptyset$.
%  Given $p$ homogeneous processors, \tricrit is the problem of minimizing the 
%  energy consumption of Equation~(\ref{eq.energy}), subject to the deadline 
%  bound~$D$ and to the reliability constraint, while using replication and/or 
%  re-execution. 
\end{itemize}
%\end{definition}

% Finally we detail here some useful notations that are used in this paper:
% \begin{itemize}
%     \item $S=\sum_{T_i \in V} w_i$, i.e., $S$ is the sum of the computation
% requirements of all tasks. 
%     \item ??
% \end{itemize}

\section{Linear chains}
\label{sec.lin}

In this section, we focus on the \chain problem, that was shown to be
NP-hard even on a single processor \cite{rr7757}. We derive an FPTAS
(Fully Polynomial Time Approximation Scheme) to sol\-ve the general
problem with replication and re-execution on $p$~processors. We start
with some preliminaries in Section~\ref{lin.char} that allow us to
characterize the shape of an optimal solution, and then we detail the
FPTAS algorithm and its proof in Section~\ref{lin.fptas}.

\subsection{Characterization}
\label{lin.char}

First, we note that while \chain is NP-hard even on a single
processor, the problem has polynomial complexity if no replication nor
re-execution can be used. Indeed, each task is executed only once, and
the energy is minimized when all tasks are running at the same speed. 
Note that this result can be found in \cite{aupy12ccpe}. 
\begin{lemma}
  \label{lemma_norel}
  Without replication or re-execution, solving \chain can be done 
  in polynomial time, and each task is
  executed at speed  $\max\left(\fr,\frac{S}{D}\right)$. 
\end{lemma}
\begin{proof}
  For a linear chain of tasks, all tasks can be mapped on the same
  processor, and scheduled following the dependencies. No task may
  start earlier by using another processor, and all tasks run at the
  same speed. Since there is no replication nor re-execution, each
  task must be executed at least at speed \fr for the reliability
  constraint. If $S/\fr>D$, then the tasks should be executed at speed
  $S/D$ so that the deadline constraint is matched (recall that
  $S=\sum_{i=1}^n w_i$), hence the result. 
\end{proof}

Next, accounting for replication and re-execution, we characterize the
shape of an optimal solution. For linear chains, it turns out that
with a single processor, only re-execution will be used, while with
more than two processors, there is an optimal solution that do not use
re-execution, but only replication. 
\begin{lemma}[Replication or re-execution]
  \label{chain_reporreex}
  When there is only one processor, it is optimal to only use
  re-execu\-tion to solve \chain.  When there are
  at least two processors, it is optimal to only use replication to
  solve \chain.
\end{lemma}

\begin{proof}
  With one processor, the result is obvious, since replication cannot
  be used. With more than one processor, if re-execution was used on
  task~$T_i$, for $1\leq i \leq n$, we can derive a solution with the
  same energy consumption and a smaller execution time by using
  replication instead of re-execution. Indeed, all instances of
  tasks~$T_j$, for $j<i$, must finish before $T_i$ starts its
  execution, and similarly, all instances of tasks~$T_j$, for $j>i$,
  cannot start before both copies of~$T_i$ has finished its
  execution. Therefore, there are always at least two processors
  available when executing~$T_i$ for the first time, and the execution
  time is reduced when executing both copies of~$T_i$ in parallel
  (replication) rather than sequentially (re-execution). 
\end{proof}

We further characterize the shape of an optimal solution by showing
that two copies of a same task can always be executed at the same
speed. 
\begin{lemma}[Speed of the replicas]
  \label{lemma.speed.chain}
  For a linear chain, when a task is executed two times, it is optimal
  to have both replicas executed at the same speed.
\end{lemma}

\begin{proof}
  The proof for re-execution has been done by \cite{rr7757}: by
  convexity of the energy and reliability functions, it is always
  advantageous to execute two times the task at the same speed, even
  if the application is not a linear chain. 

  For replication, this lemma is only true in the case of linear
  chains. Indeed, because of the structure of the chain, as explained
  in the proof of Lemma~\ref{chain_reporreex}, both copies of a task
  have the same constraints on starting and ending time, and hence it
  is better to execute them exactly at the same time. 
\end{proof}

We can further characterize an optimal solution by providing detailed
information about the execution speed of the tasks, depending whether
they are executed only once, re-executed, or replicated. 
\begin{proposition}
  \label{prop_WC_fr}
  If $D > \frac{S}{\fr}$, then in any optimal solution of
  \chain, all tasks that are neither
  re-executed nor replicated are executed at speed~\fr.  Furthermore, 
  let $V_r \subseteq V$ be the subset of tasks that are either
  re-executed or replicated. Then, these tasks are all executed at the
  same speed \freex, if $\freex \geq \max(\fmin, \max_{T_i \in V_r}
  \finf)$. 
\end{proposition}

\begin{proof}
  The proof for $p=1$ (re-execution) can be found in \cite{rr7757}.
  We prove the result for $p \geq 2$, which corresponds to the case
  with replication and no re-execution (see
  Lemma~\ref{chain_reporreex}).
  Note first that since $D > \frac{S}{\fr}$, if no task is
  replicated, we have enough time to execute all tasks at speed~\fr.

  Now, let us consider that task~$T_i$ is replicated at
  speed~$f_i$ (recall that both replicas are executed at the same
  speed, see Lemma~\ref{lemma.speed.chain}), and task~$T_j$ is executed
  only once at speed~$f_j$. Then, we have
  $f_j\geq\fr$ (reliability constraint on~$T_j$), and 
  $\frac{1}{\sqrt{2}}\fr \geq f_i$ %(the first inequality
 % is due to the reliability constraint on $T_j$, the last inequality
  %is a consequence of the energy minimization:
  (otherwise, executing $T_i$ only once at speed~\fr would improve
  both the energy and the execution time while matching the
  reliability constraint).

  If $f_j > \fr$, let us show that we can rather execute~$T_j$ at
  speed~$\fr$ and $T_i$ at a new speed $f'_i > f_i$, while keeping the
  same deadline: $\frac{w_i}{f'_i} + \frac{w_j}{\fr} = \frac{w_i}{f_i}
  + \frac{w_j}{f_j}$.  The energy consumption is then $2w_i f_i^{'2} +
  w_j \fr^2$.  Moreover, we know that the minimum of the function
  $2w_i f_1^2 + w_j f_2^2$, given that $\frac{w_i}{f_1} +
  \frac{w_j}{f_2}$ is a constant (where $f_1$ and $f_2$ are the
  unknowns), is obtained for $f_1 = \frac{1}{2^{1/3}}f_2$ (see
  Theorem~1 by \cite{aupy12ccpe}).  Therefore, if the optimal speed of
  $T_j$ (i.e., $f_2$) is strictly greater than~$\fr$, then the optimal
  speed for $T_i$ is $f'_i=f_1 = \frac{1}{2^{1/3}}f_2 >
  \frac{1}{2^{1/2}}f_2 > \frac{1}{2^{1/2}} \fr$, that means that we
  can improve both energy and execution time by executing $T_i$ only
  once at speed~\fr.  Otherwise, the speed of $T_j$ is further
  constrained by~\fr, hence the previous inequality ($f_1 =
  \frac{1}{2^{1/3}}f_2$) does not hold anymore, and the function is
  minimized for $f_2=\fr$. The value of $f'_i$ can be easily deduced
  from the constraint on the deadline.  This proves that all tasks
  that are not replicated are executed at speed~\fr.

  Let $M=\max(\fmin, \max_{T_i \in V_r}\finf)$. We now prove that if two
  tasks are replicated at a speed greater than~$M$, then both tasks
  are executed at the same speed. %Then we have that every replicated
%  tasks is executed at the same speed, which we call \freex.
  Suppose that $T_i$ and $T_j$ are executed twice at speeds $f_i > f_j
  \geq M$. Let $\tilde{f} = f_i f_j \frac{w_i+w_j}{w_if_j + w_jf_i}$.
  Then $f_i > \tilde{f} > f_j \geq M$, and therefore we can execute
  both tasks at speed~$\tilde{f}$ while keeping the same deadline and
  matching the reliability constraints. By convexity, such an
  execution gives a better energy consumption. We can iterate on all
  the tasks that are replicated, hence obtaining the speed at which
  each task will be re-executed, \freex. This concludes the proof.
\end{proof}

Following Proposition~\ref{prop_WC_fr}, we are able to precisely
define~\freex, and give a closed form expression of the energy of a
schedule. 
\begin{corollary}
  \label{cor.energy.chain}
Given a subset $V_r$ of tasks re-executed or replicated, let 
%$S=\sum_{i=1}^n w_i$,  
$X=\sum_{T_i\in V_r} w_i$, and 
\[
  \freex = \left\{
  \begin{array}{l l}
    \max\left(\fmin, \frac{2X}{D\fr-S+X}\fr\right) & \quad \text{if $p=1$};\\
    \max\left(\fmin, \frac{X}{D\fr-S+X}\fr\right) & \quad \text{if $p\geq 2$}.\\
  \end{array} \right. 
\]
Then, if $\freex\geq \max_{T_i\in V_r} \finf$,  
the optimal energy consumption is 
\begin{equation}
(S-X)\fr^2 + 2X \freex^2. 
\end{equation}

Note that the energy consumption only depends on~$X$, and therefore
\chain is equivalent in this case
to the problem of finding the optimal set of tasks that have to be
re-executed or replicated.
\end{corollary}

\begin{proof}
  Given a deadline~$D$, the problem is to find the set of tasks
  re-executed (or replicated), and the speed of each task. Thanks to
  Proposition~\ref{prop_WC_fr}, we know that the tasks that are not in
  this set are executed at speed~\fr, and given the set of tasks
  re-executed or replicated, we can easily compute the optimal speed
  to execute each task in order to minimize the energy consumption:
  all tasks are executed at the same speed, and we have
  $\lambda \frac{X}{\freex} + \frac{S-X}{\fr} = D$, with $\lambda=1$
  in the case of replication ($p\geq 2$), and $\lambda=2$ in the case
  of re-execution ($p=1$). 
  Hence the corollary.
\end{proof}

%\todo{Here I can make a small note about how it is when $\freex < \finf$ for 
%some $i$.}
\begin{remark}
  Note that if there is a task $T_i \in V_r$ such that $\finf >
  \freex$, then the optimal solution for this set of replicated tasks
  is obtained by executing~$T_i$ at speed \finf, and by executing all
  the other tasks at a new speed $\freex^{\texttt{new}} \leq \freex$,
  such that $D$ is exactly met.  We can do this recursively until
  there are no more tasks~$T_i$ such that $\finf >
  \freex^{\texttt{new}}$. Using the procedure \computeVl($V_r$) (see
  Algorithm~\ref{algo.computevl}), we can
  compute the optimal energy consumption in a time polynomial
  in~$|V_r|$. 

\begin{algorithm}[htb]
\caption{Computing re-execution speeds; tasks in $V_r$ are
  re-executed.}
\label{algo.computevl}
procedure {\computeVl}($V_r$)\\
\Begin{
$V_l^{(0)} = \emptyset$\;
$\freex^{(0)} =  \left\{
  \begin{array}{l l}
    \max\left(\fmin, \frac{2X}{D\fr-S+X}\fr\right) & \quad \text{if $p=1$};\\
    \max\left(\fmin, \frac{X}{D\fr-S+X}\fr\right) & \quad \text{if $p\geq 2$}.\\
  \end{array} \right.$\\
$j=0$\;
\While{$j=0$ or $V_l^{(j)} \neq V_l^{(j-1)}$}{
$j:= j+1$\;
$V_l^{(j)} = V_l^{(j-1)} \cup \{T_i\in V_r \;|\; \finf > \freex^{(j-1)} \}$\;
$\freex^{(j)} =  \left\{
  \begin{array}{l l}
    \max\left(\fmin, \frac{\sum_{T_i \in V_r\setminus V_l^{(j)}} 2w_i }{D-\frac{S-X}{\fr} - \sum_{T_i \in V_l^{(j)}}\frac{2w_i}{\finf}}\right) & \text{if $p=1$};\\
    \max\left(\fmin, \frac{\sum_{T_i \in V_r\setminus V_l^{(j)}} w_i }{D-\frac{S-X}{\fr} - \sum_{T_i \in V_l^{(j)}}\frac{w_i}{\finf}}\right) & \text{if $p\geq 2$}.\\
  \end{array} \right.$
}
\Return{$(V_l^{(j)}, \freex^{(j)})$;}
}
\end{algorithm}

Let $(V_l,\freex)$ be the result of \computeVl($V_r$). Then the optimal
energy consumption is
%\begin{equation}
$  (S-X)\fr^2 + \sum_{T_i \in V_l} 2w_i\finf^2 + \sum_{T_i \in V_r
  \setminus V_l} 2w_i \freex^2$ .
%\end{equation}

\end{remark}

\begin{corollary}
If $D > \frac{S}{\fr}$, \chain can be solved using an exponential time 
exact algorithm.
\end{corollary}

\begin{proof}
  The algorithm computes for every subset $V_r$ of tasks the energy
  consumption if all tasks in this subset are re-executed, and it
  chooses one with the minimal energy consumption, that corresponds to
  an optimal solution. It takes exponential time to compute every
  subset~$V_r\subseteq V$, with $|V|=n$. 
\end{proof}

Thanks to Corollary~\ref{cor.energy.chain}, we are also able to identify
problem instances that can be solved in polynomial time.

\begin{theorem}
  \label{thm.chain}
 \chain can be solved in polynomial time in the
 following cases: 
\begin{enumerate}
    \item $D \leq \frac{S}{\fr}$ (no re-execution nor replication);
    \item $p=1$, $D \geq \frac{1+c}{c}\frac{S}{\fr}$, where $c$~is the 
      only positive
      solution to the polynomial $7X^3 + 21 X^2 - 3X -1 = 0$, and
      hence $c = 4 \sqrt{\frac{2}{7}} \cos{\frac{1}{3}(\pi-\tan^{-1}{\frac{1}{\sqrt{7}}})}-1 ~ (\approx 0.2838)$, and 
      for $1\leq i\leq n$, $\finf \leq \frac{2c}{1+c}\fr$ (all
      tasks can be re-executed);
%    \item $p \geq 2$ and $D \leq \frac{S}{\fr}$;
    \item $p \geq 2$, $D \geq 2\frac{S}{\fr}$, and 
      for $1\leq i\leq n$, $\finf \leq \frac{1}{2}\fr$ (all tasks can
      be replicated).
\end{enumerate}
\end{theorem}

\begin{proof}
  First note that when $D \leq \frac{S}{\fr}$, the optimal solution is
  to execute each task only once, at speed $\frac{S}{D}$, since $S/D
  \geq \fr$. Indeed, this solution matches both reliability and
  makespan constraints, and it was proven to be the optimal solution
  in Proposition~2 by \cite{aupy12ccpe} (it is easy to see that
  replication or re-execution would only increase the energy
  consumption).

  Let us now consider that $D>\frac{S}{\fr}$. We aim at showing
  that the minimum of the energy function is reached when the total
  weight of the re-executed or replicated tasks is
\[
\left\{
  \begin{array}{l l}
     c(D\fr - S)   & \quad \text{if $p=1$};\\
     (D\fr- S)     &  \quad \text{if $p\geq 2$}.\\
  \end{array} \right.
\]
Then necessarily, when this total weight is greater than~$S$, the
optimal solution is to re-execute or replicate all the tasks. Hence
the theorem. We differentiate the two cases in the following ($p=1$ or
$p=2$).

\paragraph{Case 1 ($p=1$). }
We want to show that the minimum energy is reached when the total
weight of the subset of tasks is exactly $c(D\fr -S)$. Let $I=\{i\;
|\; T_i$ is executed twice in the solution$\}$, and let
$X=\sum_{i\in I}a_i$.

We saw in Corollary~\ref{cor.energy.chain} that the energy consumption
cannot be lower than $(S-X)\fr^2 + 2X \freex^2$ where $\freex =
\frac{2X}{D\fr-S+X}\fr$.  Therefore, we want to minimize $E(X) =
(S-X)\fr^2 + 2X \left( \frac{2X}{D\fr-S+X}\fr \right)^2$.

If we differentiate $E$, we can see that the minimum is reached when 
$-1 + \frac{24X^2}{(D\fr - S + X)^2} - \frac{16X^3}{(D\fr - S + X)^3}
= 0$,
 that is, 
 $-(D\fr - S + X)^3 + 24X^2(D\fr - S + X) - 16X^3 = 0$, or 
\begin{align*}
7X^3 +& 21(D\fr -S)X^2  - 3(D\fr -S)^2X -(D\fr -S)^3 = 0.
\end{align*} 
The only positive solution to this equation is $X = c(D\fr -S)$, and
therefore the minimum is reached for this value of $X$, and
then $\freex=\frac{2c}{1+c}\fr$. 

When $X\geq S$, re-executing each task is the best strategy to
minimize the energy consumption, and that corresponds to the case $D
\geq \frac{1+c}{c}\frac{S}{\fr}$. The re-execution speed may then be
lower than $\frac{2c}{1+c}\fr$. Therefore, it may happen that $\finf >
\freex$ for some task~$T_i$. However, even with a tighter deadline, it
would be better to re-execute~$T_i$ at speed $\frac{2c}{1+c}\fr$
rather than to execute it only once at speed~\fr. Therefore, since
$\finf \leq \frac{2c}{1+c}\fr$, it is optimal to re-execute~$T_i$, at
the lowest possible speed, i.e.,~\finf. Note that this changes the
value of \freex, and the call to \mbox{\computeVl($V$)} (see
Algorithm~\ref{algo.computevl}) returns tasks that are executed
at~\finf, together with the re-execution speed for all the other tasks.

% Note that there might be a limitation when there is a tasks $T_i$ that should
% be re-executed at a speed lower than \finf. However, because 
% $\finf<\frac{2c}{1+c}\fr$ (that is the optimal speed when the re-executed weight
% is equal to $X$), the convexity of the energy functions tells us that it is 
% better to re-execute $T_i$ at speed \finf. In that case, \computeVl called on 
% $V$ will give the optimal solution by convexity.
% \todo{improve here.}

\paragraph{Case 2 ($p\geq 2$). }
Similarly, we want to show that, in this case, the minimum energy is
reached when the total weight of the subset of tasks that are
replicated is exactly $D\fr-S$. Let $I=\{i\; |\; T_i$ is executed
twice in the solution$\}$, and let $X=\sum_{i\in I}a_i$.

We saw in Corollary~\ref{cor.energy.chain} that the energy consumption
cannot be lower than $(S-X)\fr^2 + 2X \freex^2$ where $\freex =
\frac{X}{D\fr-S+X}\fr$.  Therefore, we want to minimize $E(X) =
(S-X)\fr^2 + 2X \left( \frac{X}{D\fr-S+X}\fr \right)^2$.

If we differentiate $E$, we can see that the minimum is reached when 
$$-1 + \frac{6X^2}{(D\fr - S + X)^2} - \frac{4X^3}{(D\fr - S + X)^3} =
0,$$
that 
is, $-(D\fr - S + X)^3 + 6X^2(D\fr - S + X) - 4X^3 = 0$, or 
\begin{align*}
X^3 +& 3(D\fr -S)X^2 - 3(D\fr -S)^2X -(D\fr -S)^3 = 0.
\end{align*}
The only positive solution to this equation is $X = D\fr -S$, and
therefore the minimum is reached for this value of~$X$, and then 
$\freex=\frac{1}{2}\fr$.  

When $X\geq S$, replicating each task is the best strategy to 
minimize the energy consumption, and that corresponds to the case 
$D \geq \frac{2S}{\fr}$. Similarly to Case~1, it is easy to see that
each task should be replicated, even if $\finf>\freex$, since \mbox{$\finf
\leq \frac{1}{2}\fr$.} The optimal solution can also be obtained with a
call to \computeVl($V$). 
% Again, there might be a limitation when there is a tasks $T_i$ that should
% be re-executed at a speed lower than \finf. However, because 
% $\finf<\frac{1}{2}\fr$ (that is the optimal speed when the re-executed weight
% is equal to $X$), the convexity of the energy functions tells us that it is 
% better to re-execute $T_i$ at speed \finf. In that case, \computeVl called on 
% $V$ will give the optimal solution by convexity.
\end{proof}

\subsection{FPTAS for \chain}
\label{lin.fptas}

We derive in this section a fully polynomial time approximation scheme
(FPTAS) for \chain, based on the FPTAS for SUBSET-SUM \cite{cormen},
and the results of Section~\ref{lin.char}. Without loss of generality,
we use the term {\em replication} for either re-execution or
replication, since both scenarios have already been clearly
identified. The problem consists in identifying the set of replicated
tasks~$V_r$, and then the optimal solution can be derived from
Corollary~\ref{cor.energy.chain}; it depends only on the total weight
of these tasks, $\sum_{T_i\in V_r} w_i$, denoted in the following
as~$w(V_r)$.

Note that we do not account in this section for \finf or \fmin for
readability reasons: \finf can usually be neglected because $\lambda_0
w_i/f$ is supposed to be very small whatever~$f$, and \fmin simply
adds subcases to the proofs (rather than an execution at speed~$f$,
the speed should be $\max(f,\fmin)$).

\medskip
%\noindent {\bf Preliminary functions.} 
First we introduce a few preliminary functions in
Algorithm~\ref{algo.prel}, and we exhibit their properties. These are
the basis of the approximation algorithm. 

%\begin{lemma}
%    \label{lemma.xopt}
When $D > \frac{S}{\fr}$, \xopt($V,D,p$) returns the optimal value for
the weight~$w(V_r)$ of the subset of replicated tasks~$V_r$, i.e.,
the value that minimizes the energy consumption for \chain. 
%, which is equivalent to a lower 
%bound on the energy consumption, for the problem \chain.
%\end{lemma}
The optimality comes directly from the proof of Theorem~\ref{thm.chain}. 
%Thanks to Theorem~\ref{thm.chain}, we are able to compute the optimal
%value for the weight~$w(V_r)$, where $V_r$ is the set of replicated
%tasks: 

%\begin{proof}
%This is a direct corollary from the proof of Theorem~\ref{thm.chain}.
%\end{proof}

Given a value~$X$, which corresponds to $w(V_r)$, 
\energy($V,D,p,X$) returns the optimal energy consumption when a
subset of tasks~$V_r$ is replicated. 

Then, the function \trim($L,\varepsilon,X$) trims a sorted list
$L=[L_0, \cdots, L_{m-1}]$ in time~$O(m)$, given $L$ and
$\varepsilon$. $L$ is sorted into non decreasing order. The function
returns a trimmed list, where two consecutive elements differ from at
least a factor $(1+\varepsilon)$, except the last element, that is the
smallest element of~$L$ strictly greater than~$X$. This trimming
procedure is quite similar to that used for SUBSET-SUM \cite{cormen},
except that the latter keeps only elements lower than~$X$. Indeed,
SUBSET-SUM can be expressed as follows: given $n$ strictly positive
integers $a_1, \ldots, a_n$, and a positive integer~$X$, we wish to
find a subset $I$ of $\{1,\ldots, n\}$ such that \mbox{$\sum_{i\in I}a_i$}
is as large as possible, but not larger than~$X$.  
 In our case, the optimal solution may be
obtained either by approaching~$X$ by below or by above.   
% output of the procedure is a pseudo-trimmed, sorted
%list. It is pseudo-trimmed because every element greater than $X$ are
%kept. We will see later why this is useful.

%Given the procedure \trim, we can construct our approximation scheme
%as follows.  This procedure
Finally, the approximation algorithm is \approxchain$(V,\ab D,p,\varepsilon)$
(see Algorithm~\ref{algo.prel}),
where % takes as input a set $V$ of tasks, a
%deadline $D$, and an "approximation parameter" $\varepsilon$, where
\mbox{$0<\varepsilon<1$}, and it returns an energy consumption~$E$ that is
not greater than \mbox{($1+\varepsilon$)} times the optimal energy
consumption. 
%It returns a value $E$ whose value is within $1+\varepsilon$ factor of the optimal solution.
Note that if $L=[L_0,\ldots,L_{m-1}]$, then \add($L,x$) adds
element~$x$ at the end of list~$L$ (i.e., it returns the list
$[L_0,\ldots,L_{m-1},x]$); $L+w$ is the list
$[L_0+w,\ldots,L_{m-1}+w]$; and \merge($L,L'$) is merging two sorted
lists (and returns a sorted list).

\begin{algorithm}
\caption{Approximation algorithm for \chain.}
\label{algo.prel}
function {\xopt}($V,D,p$)\\
\Begin{
$S=\sum_{T_i \in V} w_i$\;
\lIf{$p=1$}{\Return{$c(D\fr-S)$}\;}
\lElse{\Return{$D\fr-S$}\;}
}

function {\energy}($V,D,p,X$)\\
\Begin{
$S=\sum_{T_i \in V} w_i$\;
\lIf{$p\!=\!1$}{\Return{$(S\!-\!X)\fr^2 \!+\! 2X\!\left(\max \left(\fmin, \frac{2X}{D\fr-S+X}\fr\right)\right)^2$}\!\!\;}
\lElse{\Return{$ (S-X)\fr^2 + 2X\left(\max \left(\fmin, \frac{X}{D\fr-S+X}\fr\right)\right)^2$}\;}
}

function {\trim}($L,\varepsilon,X$)\\
\Begin{
$m = |L|$; %.\text{length}$; 
$L=[L_0,\ldots,L_{m-1}]$; 
$L'= [L_0]$;
$last = L_0$\;
\For{$i=1$ \KwTo $m-1$}{
\If{ ($last\leq X$ and $L_i > X$) or $L_i > last\times (1+\varepsilon)$}
%\tcc*{%$L_i \geq last$ since 
%$L$ is sorted}
{$L' = \add(L',L_i)$;
$last = L_i$\;}
}
\Return{$L'$;} 
}

%\end{algorithm}

%\begin{algorithm}[H]
%\caption{\approxchain($V,D,p,\varepsilon$)}
function \approxchain($V,D,p,\varepsilon$)\\
\Begin{
$X = \lfloor \xopt(V,D,p) \rfloor$; 
$n=|V|$;
$L^{(0)} = [0]$\;
    \For{$i=1$ \KwTo $n$}{
$L^{(i)} = \merge(L^{(i-1)}, L^{(i-1)}+w_i)$\;
$L^{(i)} = \trim(L^{(i)},\varepsilon / (28\times 2n), X)$\;
%remove from $L_i$ every element that is greater than $X$ but the smallest of them \;
}
Let $Y_1 \leq Y_2$ be the two largest elements of $L^{(n)}$;  \\
    \Return{$\min(\energy(V,D,p,Y_1),\energy(V,D,p,Y_2))$;}
}
\end{algorithm}

We now prove that this approximation scheme is an FPTAS: 

\begin{theorem}
    \label{approx.chain}
    \approxchain\xspace is a fully polynomial time approximation
    scheme for \chain.
\end{theorem}

\begin{proof}
We assume that
\begin{compactitem}
    \item if $p=1$, then $\frac{S}{\fr}<D< \frac{1+c}{c}\frac{S}{\fr}
      <5\frac{S}{\fr}$; 
    \item if $p\geq2$, then $\frac{S}{\fr}<D< 2\frac{S}{\fr}$; 
\end{compactitem}
otherwise the optimal solution is obtained in polynomial time (see
Theorem~\ref{thm.chain}).  

% First, we prove some preliminary results, based on the relation
% $\leq_{V}$ defined below: 
% \begin{definition}
%   Given $V_1 \subseteq V$ and $V_2 \subseteq V$, we define the relation
%   $\leq_{V}$ such that
% $V_1 \leq_V V_2$ if and only if  
% $\sum_{T_i \in V_1} w_i \leq \sum_{T_i \in V_2} w_i$. 
% If  $V_1 \leq_V V_2$ and $V_2 \leq_V V_1$, then $V_1 =_V V_2$.
% \end{definition}

% Note that $V_1 =_V V_2$ does not imply $V_1=V_2$. 
% %; this relation is not antisymmetric.

% \begin{lemma}
%  \label{lemma.lin.eq}
%  Let $V_1\subseteq V$ and $V_2 \subseteq V$ such that $V_1 =_V V_2$. Then
%  the optimal energy consumption when $V_1$ is the set of 
%  replicated tasks is equal to the one when $V_2$ is this set.
% \end{lemma}

% \begin{proof}
%   This is an easy corollary from the proof of
%   Proposition~\ref{prop_WC_fr}, since the energy consumption of a
%   linear chain is not dependent on the number of tasks replicated, but
%   on the sum of their weights.
% \end{proof}

%\begin{definition}
%We define the subsets of~$V$: 
%\begin{itemize}
%    \item 
Let $I_{\inf} = \{V' \subseteq V\; |\; w(V') \leq
      \xopt(V,D,p) \}$, and 
%    \item 
$I_{\sup} = \{V'' \subseteq V\; |\; w(V'') >
      \xopt(V,D,p) \}$. 
%\end{itemize}
%\end{definition}
Note that $I_{\inf}$ is not empty, since $\emptyset \in I_{\inf}$. 
%Note that neither $I_{\inf}$ nor $I_{\sup}$ is empty.  Indeed,
%$\emptyset \in I_{\inf}$ and $V \in I_{\sup}$, since $0 < \xopt(V,D,p)
%< S$.

\smallskip
First we characterize the solution with the following lemma: 

\begin{lemma}
    \label{lemma.i1i2}
    Suppose $D > \frac{S}{\fr}$.  Then in the solution of \chain, the
    subset of replicated tasks~$V_r$ is either an element
    $V'\in I_{\inf}$ such that $w(V')$ is maximum, or an element
    $V''\in I_{\sup}$ such that $w(V'')$ is minimum. 
%one of the maximum (according
%    to $\leq_V$) elements of $I_{\inf}$, or one of the minimum
%    (according to $\leq_V$) elements of $I_{\sup}$.
\end{lemma}

\begin{proof}
%  First we saw in Lemma~\ref{lemma.lin.eq} that every maximum
%  (according to $\leq_V$) element of $I_{\inf}$ will give the same
%  minimal energy.
  Recall first that according to Proposition~\ref{prop_WC_fr}, the
  energy consumption of a linear chain is not dependent on the number
  of tasks replicated, but only on the sum of their weights. 

  Then the lemma is obvious by convexity of the functions, and 
     %because we saw in Lemma~\ref{lemma.xopt} that 
  since \xopt\xspace returns the
  optimal value of~$w(V_r)$, the weight of the replicated tasks.
  Therefore, the closest the weight of the set of replicated tasks is
  to the optimal weight, the better the solution is. Finally, % it
%  suffices to say that 
  any element in $I_{\inf}$ is a solution (since we have a solution
  for \xopt), and if the minimal element (if it exists) of $I_{\sup}$
  is not a solution, (\freex too large because of time constraints),
  then no element of $I_{\sup}$ can be a better solution. 
\end{proof}

%Denote $S = \sum w_i$.

We are now ready to prove Theorem~\ref{approx.chain}.  
Let $X_1=\max_{V_1\in I_{\inf}} w(V_1)$, and 
 $X_2=\max_{V_2\in I_{\sup}} w(V_2)$.  
Thanks to Lemma~\ref{lemma.i1i2}, the optimal set of replicated
tasks~$V_o$ is such that $X_o=w(V_o)=X_1$ or $X_o=X_2$. 
The corresponding energy consumption is (Corollary~\ref{cor.energy.chain}): 
% \sum_{T_i \in V_2} w_i$
%such that $V_2$ is a minimum ($\leq_V$) element of $I_{\sup}$. We
%define $X_o$ such that the minimal energy consumption is
\[
  E_{opt} = \left\{
  \begin{array}{l l}
    (S-X_o)\fr^2 + \frac{(2X_o)^3}{(D\fr-S+X_o)^2}\fr^2 & \text{if $p=1$}\\
    (S-X_o)\fr^2 + \frac{2X_o^3}{(D\fr-S+X_o)^2}\fr^2 & \text{if $p\geq 2$}\\
  \end{array} \right.  . 
\]
%$E_{opt} = (S-X_o)\fr^2 + \frac{2X_o^3}{(D\fr-S+X_o)^2}\fr^2$ 
%(Corollary~\ref{cor.energy.chain}). We know thanks to Lemma~\ref{lemma.i1i2} 
%that $X_o = X_1$ or $X_2$.

The solution returned by \approxchain\xspace %, of weight~$X_a$,
corresponds either to~$Y_1$ or to~$Y_2$, where $Y_1$ and~$Y_2$ are the
two largest elements of the trimmed list. We first prove that at least
one of these two elements, denoted~$X_a$, is such that  
$X_a \leq X_o \leq (1+\varepsilon')X_a$, where $\varepsilon'=
\frac{\varepsilon}{28}$. 

% We first show the existence of a
% $X_a$ (amongst $\{Y_1,Y_2\}$ computed by \approxchain), such that $X_a
% \leq X_o \leq (1+\varepsilon')X_a$ (where $\varepsilon'=
% \frac{\varepsilon}{28}$). Then we show that the energy $E_{algo}$
% obtained with this value~$X_a$ is such that $E_{opt}\leq E_{algo} \leq
% (1+\varepsilon)E_{opt}$. Since the energy computed by the algorithm is
% smaller or equal to this energy (we take the minimum of both), we have
% our result.

\paragraph{Existence of $X_a$ such that $X_a \leq X_o \leq (1+\varepsilon')X_a$. }
%Our algorithm is an adaptation from the FPTAS~\cite{cormen} for the
%SUBSET-SUM problem (applied on $A=\{w_i\; |\; T_i \in V\},~X = \lfloor
%\xopt(V,D,p) \rfloor$). 
We differentiate two cases. 

\begin{description}
\item[(a)] If $Y_2 > X$, then $Y_1$ is the value obtained by the FPTAS
  for SUBSET-SUM \cite{cormen} with the approximation
  ratio~$\varepsilon'$, since it is the largest value not greater
  than~$X$, and our algorithm is identical for such values.  Moreover,
  note that $X_1$~is the optimal solution of SUBSET-SUM by definition,
  and therefore $Y_1 \leq X_1<(1+\varepsilon')Y_1$.  If $X_o=X_1$, the
  value $X_a=Y_1$ satisfies the property.  

  If $X_o=X_2$, we prove that the property remains valid, by
  considering the SUBSET-SUM problem with a bound~$X_2$ instead
  of~$X$. Then, since $Y_2>X$, we have $Y_2\geq X_2$ by definition
  of~$X_2$. Moreover, \approxchain\xspace is not removing any element
  of the list greater than~$Y_2$, and therefore all elements between
  $X$ and~$X_2$ are kept, similarly to the FPTAS for SUBSET-SUM. 
  If $Y_2=X_2$, then $X_a=Y_2$ satisfies the property. Otherwise,
  $Y_1$ is the result of the FPTAS for SUBSET-SUM with a bound~$X_2$,
  whose optimal solution is~$X_2$, and therefore $Y_1$ is such that 
  $Y_1 \leq X_2 <(1+\varepsilon')Y_1$; $X_a=Y_1$ satisfies the
  property.  

% exactly the result of the
%   FPTAS applied on $(A,X)$. Since $X_1$ is the optimal solution for
%   the SUBSET-SUM problem for $X$ (by definition of $I_{\inf}$), we
%   have $Y_1 \leq X_1<(1+\varepsilon')Y_1$. Then let us show that we
%   have either $Y_2 = X_2$, either $Y_1 \leq X_2<(1+\varepsilon')Y_1$.
%   Suppose $Y_2 > X$. Then because we know that there is no subset of
%   $V$ that sums up to an integer between $X$ and $X_2$ by definition
%   of $X_2$, then either $Y_2 = X_2$ (then $Y_2 \leq
%   X_2<(1+\varepsilon')Y_2$), either $Y_2 > X_2$. In the later case,
%   because we know that we never removed from the list $L_i$ any value
%   between $X$ and $X_2$ (all the value we removed were greater than
%   $Y_2$ by construction of $Y_2$). Hence the approximation algorithm
%   for $X$ is also valid for $X_2$, hence: $Y_1 \leq
%   X_2<(1+\varepsilon')Y_1$.

\item[(b)] If $Y_2 \leq X$, no elements greater than~$X$ have been
  removed from the lists, and \approxchain\xspace has been identical
  to the FPTAS for SUBSET-SUM. Then, $X_a=Y_2$ is the solution, that
  is valid both for SUBSET-SUM applied with the original bound~$X$
  (optimal solution~$X_1$), and with the modified bound $X_2$ (optimal
  solution~$X_2$). Therefore, $Y_2 \leq X_1<(1+\varepsilon')Y_2$ and $Y_2 \leq
  X_2<(1+\varepsilon')Y_2$, which concludes the proof. 
% this mean that we never had to remove any
%   element from $L_i$ which is the only part in the SUBSET-SUM approx
%   that depended on $X$. So our approx is not only valid for $X_1$ the
%   solution of SUBSET-SUM, but also valid for any integer greater than
%   $X$.  Then $Y_2 \leq X_1<(1+\varepsilon')Y_2$ and $Y_2 \leq
%   X_2<(1+\varepsilon')Y_2$.
\end{description}

We have shown that there always is $X_a$ (either $Y_1$ or $Y_2$) such
that $X_a\leq X_o <(1+\varepsilon')X_a$.  Next, we show that the
energy $E_a$ obtained with this value~$X_a$ is such that
$E_{opt}\leq E_a \leq (1+\varepsilon)E_{opt}$.

\paragraph{Approximation ratio on the energy:  %$E_{opt}\leq E_a
 $E_a \leq (1+\varepsilon)E_{opt}$.} 

%Because $E_{algo} = \min (\energy(Y_1),\energy(Y_2))$, we know that necessarily, 
%$\frac{E_{algo}}{\fr^2} \leq  \left\{
%  \begin{array}{l l}
%    S-X_a + \frac{(2X_a)^3}{(D\fr-S+X_a)^2} & \quad \text{if $p=1$}\\
%    S-X_a + \frac{2X_a^3}{(D\fr-S+X_a)^2} & \quad \text{if $p\geq 2$}\\
%  \end{array} \right.$.
%Let us show the result when $p\geq 2$, the case for $p=1$ is identical.

Let us consider first that $p\geq 2$. Then we have 
$E_a =   (S-X_a)\fr^2 + \frac{2X_a^3}{(D\fr-S+X_a)^2}\fr^2$.  
Re-using the previous inequalities on $X_a$, we obtain: 
$\frac{E_{a}}{\fr^2} \leq  S-\frac{X_o}{1+\varepsilon'}  +
\frac{2X_o^3}{(D\fr-S+\frac{X_o}{1+\varepsilon'})^2}$. 
Then, this can be rewritten so that $E_{opt}$ appears: 
\begin{align*}
\frac{E_{a}}{\fr^2} 
&\leq \left(\frac{1}{1+\varepsilon'}(S-X_o)+\frac{\varepsilon'}{1+\varepsilon'}S \right ) \\
& \quad {}+\left (
  (1+\varepsilon')^2\frac{2X_o^3}{((1+\varepsilon')(D\fr-S)+X_o)^2}
\right )
\end{align*}
\begin{align*}
\frac{E_{a}}{\fr^2} 
&\leq \left((S-X_o) + \varepsilon' S \right ) \\
& \quad{} + \left((1+\varepsilon')^2\frac{2X_o^3}{(D\fr-S+X_o)^2} \right ) \\
&\leq \left((S-X_o) + \varepsilon' S \right ) \\
& \quad{} + \left( (1+\varepsilon')^2(\frac{E_{opt}}{\fr^2} - (S-X_o)) \right )\\
&\leq (1+\varepsilon')^2 \frac{E_{opt}}{\fr^2} \\
& \quad{}- ((1+\varepsilon')^2 - 1 )(S-X_o) + \varepsilon' S \\
&\leq (1+\varepsilon')^2 \frac{E_{opt}}{\fr^2} + \varepsilon' S . 
%&\leq (1+\varepsilon)^2 \frac{E_{opt}}{\fr^2} + \varepsilon' \frac{25E_{opt}}{\fr^2}
\end{align*}

The case $p=1$ leads to the same inequality; the only difference is in
the energy~$E_a$, where $2X_a^3$ is replaced by $(2X_a)^3$, and the
same difference holds for $E_{opt}$ ($2X_o^3$ is replaced by
$(2X_o)^3$). 

Finally, note that with no reliability constraints, each task is
executed only once at speed $S/D$, and therefore the energy
consumption is at least $E_{opt}\geq S\frac{S^2}{D^2}$. Moreover, by
hypothesis, $D<\frac{5S}{\fr}$ (for $p\geq 1$). Therefore,
$S<\frac{25E_{opt}}{\fr^2}$ and $\frac{E_{a}}{\fr^2} < 
(1+\varepsilon')^2 \frac{E_{opt}}{\fr^2} + \varepsilon'
\frac{25E_{opt}}{\fr^2}$.
% The last inequality comes from the fact that $E_{opt} > S\frac{S^2}{D^2}$: 
% $S\frac{S^2}{D^2}$ is the minimal energy without reliability constraints 
% (one execution at speed $\frac{S}{D}$ for all tasks). Moreover, 
% $D<\frac{5S}{\fr}$ (when $p\geq 1$). Then rewriting everything we have 
% $S<\frac{25E_{opt}}{\fr^2}$.

\medskip
We conclude that 
\begin{equation*}
\frac{E_{a}}{E_{opt}}< 1 + 27 \varepsilon' + \varepsilon'^2 < 1+28\varepsilon' = 1+\varepsilon.
\end{equation*}
%We conclude by stating that $E_{opt} \leq E_{a}$ by definition. 
%Furthermore, we know that $E_{opt} \leq E_{algo}$ (by definition).

\paragraph{Conclusion. }
The energy consumption returned by \approxchain\xspace, denoted as
$E_{algo}$, is such that $E_{algo} \leq E_a$, since we take the minimum
out of the consumption obtained for $Y_1$ or~$Y_2$, and $X_a$ is
either $Y_1$ or~$Y_2$. Therefore, $E_{algo}\leq
(1+\varepsilon)E_{opt}$. 

It is clear that the algorithm is polynomial both in the size of the
instance and in~$\frac{1}{\varepsilon}$, given that the trimming
function and \approxchain\xspace have the same complexity as in the
original approximation scheme for SUB\-SET-SUM (see~\cite{cormen}), and
all other operations are polynomial in the problem size (\xopt,
\energy). 
%Furthermore it takes polynomial time in the size of the instance since
%computing $\lfloor \xopt \rfloor$ is polynomial in the time of the
%instance, same for the \energy. The rest of the algorithm is
%polynomial in the size of the instance and in $\frac{1}{\varepsilon}$
%(see the proof of the FPTAS for SUBSET-SUM by \cite{cormen}). 
\end{proof}

%%%%%%%%%%%%%%%%%%%%%%%%%%%%%%%%%%%%%%%%%%%%%%%%%%%%%%%%%%%%%%%%%%%%%%%%%%%%%%%%
%%%%%%%%%%%%%%%%%%%%%%%%%%%%%%%%%%%%%%%%%%%%%%%%%%%%%%%%%%%%%%%%%%%%%%%%%%%%%%%%
%\newpage
%\clearpage
\section{Independent tasks}
\label{sec.indep}
    
In this section, we focus on the problem of scheduling independent
tasks, \indep. Similarly to \chain, we know that \indep is NP-hard,
even on a single processor. We first prove in
Section~\ref{sec.inapprox} that there exists no constant factor
approximation algorithm for this problem, unless P=NP. We discuss and
characterize solutions to \indep in Section~\ref{char.indep}, while
highlighting the intrinsic difficulty of the problem. The core result
is a constant factor approximation algorithm with a relaxation on the
constraint on the makespan (Section~\ref{algo.indep}).

%Throughout this section, given a schedule, the \emph{load of processor
%  $p_i$}, $l(p_i)$, is defined as the sum of the weights of the
%different tasks executed on that processor.

\subsection{Inapproximability of \indep}
\label{sec.inapprox}

\begin{lemma}
    \label{indep.inapprox}
For all $\lambda >1$, there does not exist any $\lambda$-approxi\-mation
of \indep, unless $P=NP$. 
\end{lemma}

\begin{proof}
Let us assume that there is a $\lambda$-approxi\-mation algorithm for
\indep. 
%Suppose we have an algorithm $A$ which runs in polynomial time and is a 
%$\lambda$-approxi\-mation for \indep, for some $\lambda$.
%Let us show that this algorithm would solve \partition (well known to be 
%NP-complete~\cite{GareyJohnson}), hence showing $P=NP$.
We consider an instance $\II_1$ of \partition: given $n$ strictly positive 
integers $a_1, \ldots, a_n$, does there exist a subset $I$ of $\{1,\ldots, n\}$ 
such that $\sum_{i\in I}a_i= \sum_{i \notin I} a_i$? Let
$S=\sum_{i=1}^n a_i$.

%\smallskip
We build the following instance~$\II_2$ of our problem. We have $n$ 
independent tasks~$T_i$ to be mapped on $p=2$ processors, and: 
\begin{compactitem}
\item task~$T_i$ has a weight $w_i=a_i$;
\item $\fmin= \fr = \fmax = S/2$; 
\item $D = 1$.
\end{compactitem}

\medskip
We use the $\lambda$-approxi\-mation algorithm to solve~$\II_2$, and the
solution of the algorithm $E_{algo}$ is such that $E_{algo}\leq
\lambda E_{opt}$, where $E_{opt}$ is the optimal solution. 
We consider the two following cases.\\ 
(i) If the $\lambda$-approxi\-mation algorithm returns a solution, then
necessary all tasks are executed exactly once at speed~\fmax, since
$\sum_{i=1}^n w_i/\fmax = 2$ and there are two processors. Moreover,
because of the makespan constraint, the load on each processor is
equal. Let $I$ be the indices of the tasks executed on the first
processor. We have $\sum_{i\in I}a_i= \sum_{i \notin I} a_i$, and
therefore $I$~is also a solution to~$\II_1$. \\
(ii)~If  the $\lambda$-approxi\-mation algorithm does not return a
solution, then there is no solution to~$\II_1$. Otherwise, if $I$ is a
solution to~$\II_1$, there is a solution to~$\II_2$ such that tasks
of~$I$ are executed on the first processor, and the other tasks are
executed on the second processor. Since $E_{algo}\leq \lambda E_{opt}$,
the approximation algorithm should have returned a valid solution.  

Therefore, the result of the algorithm for~$\II_2$ allows us to
conclude in polynomial time whether there is a solution to the
instance~$\II_1$ of \partition or not. Since \partition is
NP-complete \cite{GareyJohnson}, the inapproximability result is true
unless P=NP. 
\end{proof}

\subsection{Characterization}
\label{char.indep}

As discussed in Section~\ref{sec.intro}, the problem of scheduling
independent tasks is usually close to a problem of load balancing, and
can be efficiently approximated for various mono-criterion versions of
the problem (minimizing the makespan or the energy, for instance).  
However, the tri-criteria problem turns out to be much harder, and
cannot be approximated, as seen in Section~\ref{sec.inapprox}, even
when reliability is not a constraint. 

Adding reliability further complicates the problem, since we no longer
have the property that on each processor, there is a constant
execution speed for the tasks executed on this processor. Indeed, some
processors may process both tasks that are not replicated (or
re-executed), hence at speed~\fr, and replicated tasks at a slower
speed.  Similarly to Section~\ref{lin.fptas}, we use the term {\em
  replication} for either re-execution or replication; if a task is
replicated, it means it is executed two times, and it appears two
times in the load of processors, be it the same processor or two
distinct processors.

Furthermore, contrary to the \chain problem, we do not always have the
same execution speed for both executions of a task, as in
Lemma~\ref{lemma.speed.chain}:  
\begin{proposition}
\label{prop.indep}
  In an optimal solution of \indep, if a task~$T_i$ is executed twice:
\begin{compactitem}
    \item if both executions are on the same processor, then both are executed 
at the same speed, lower than $\frac{1}{\sqrt{2}}\fr$;
\item however, when the two executions of this task are on distinct
  processors, then they are not necessarily executed at the same
  speed.  Furthermore, one of the two speeds can be greater than
  $\frac{1}{\sqrt{2}}\fr$.
\end{compactitem}
Moreover, we have $w_i<\frac{1}{\sqrt{2}}D\fr$. 
\end{proposition}

\begin{proof}
  We start by proving the properties on the speeds.  
  When both executions occur on the same processor, this property was
  shown by \cite{rr7757}: a single execution at speed~\fr leads to a
  better energy consumption (and a lower execution time). 

  In the case of distinct processors, we give an example in which
  the optimal solution uses different speeds for a replicated task,
  with one speed greater than $\frac{1}{\sqrt{2}}\fr$. Note that one
  of the speeds is necessary lower than $\frac{1}{\sqrt{2}}\fr$,
  otherwise a solution with only one execution of this task at
  speed~\fr would be better, similarly to the case with re-execution. 

  Consider a problem instance with two processors, $\fr=\fmax$,
  $D=\frac{6.4}{\fmax}$, and three tasks such that $w_1 = 5$, $w_2=3$,
  and $w_3=1$.  Because of the time constraints, $T_1$ and $T_2$ are
  necessarily executed on two distinct processors, and neither of
  them can be re-executed on its processor.  The problem consists in
  scheduling task~$T_3$ to minimize the energy consumption. There are
  three possibilities: 
\begin{compactitem}
\item $T_3$ is executed only once on any of the processors, at
  speed~$\fr=\fmax$; 
%; it consumes $\fr^2=\fmax^2$;
\item $T_3$ is executed twice on the same processor; it is executed on
  the same processor than~$T_2$, hence having an execution time of
  $D-\frac{w_2}{\fmax} = \frac{3.4}{\fmax}$, and therefore both
  executions are done at a speed $\frac{2}{3.4}\fmax$; 
%(then it should be
%  the one where $T_2$ is executed because the time remaining is
%  longer), that is two times at a speed
%  $\frac{2w_3}{D-\frac{w_2}{\fmax}}$;
    \item $T_3$ is executed once on the same processor than~$T_1$ at
      a speed $\frac{1}{1.4}\fmax$, and once on the other processor at
      a speed $\frac{1}{3.4}\fmax$. 
% (at a speed 
%$\frac{w_3}{D-\frac{w_1}{\fmax}}>\frac{1}{\sqrt{2}}\fr$), and once on the 
%processor of $T_2$ at a speed $\frac{w_3}{D-\frac{w_2}{\fmax}}$.
\end{compactitem} 
It is easy to see that the minimum energy consumption is obtained with
the last solution, and that $\frac{1}{1.4}\fmax >
\frac{1}{\sqrt{2}}\fr$, hence the result.  

Finally, note that since at least one of the executions of the task
should be at a speed lower than $\frac{1}{\sqrt{2}}\fr$, and since the
deadline is~$D$, in order to match the deadline, the weight of the
replicated task has to be strictly lower than
$\frac{1}{\sqrt{2}}D\fr$. 
\end{proof}

Because of this proposition, usual load balancing algorithms are
likely to fail, since processors handling only non-replicated tasks
should have a much higher load, and speeds of replicated tasks may be
very different from one processor to another in the optimal solution. 

% \begin{comment}

% Those two propositions give good hints why we could expect the ``usuals'' 
% algorithms not to work. Indeed, we can expect the load of the processors from 
% $\platform_1$ to be greater than the loads from processors of $\platform_3$ by 
% at least a factor $\sqrt{2}$, with the load of processors from $\platform_2$ 
% somewhere inbeetween. Finding the right norm such that those load would be 
% balanced seems hard and not natural.

% \begin{itemize}
%     \item Some algorithms (see~\cite{Alon97,RenaudGoudGreedy,Graham69}), aim to minimizing the $\ell_p$ norm. 
%     \item Good solution, and equivalent to the problem.
%     \item However not applicable here for two reasons: (i) they do not take into account
% \fmax; (ii) what is the right norm with the reliability constraint (either one execution at speed $f \in [\fr,\fmax]$, either two executions at speed $f_1$, $f_2$, with $f_1 < \frac{1}{\sqrt{2}}\fr$.
% \end{itemize}
% Trying to find the right norm, Prop~\ref{prop.indep}! Then we can see that some 
% processors should have a load twice as big as the other ones. 
% Not necessarily balanced $\rightarrow$ no greedy

% Consequence of \fmax, prop inapproximability (should we say it right now?)

% \end{comment}

%    \subsubsection{Finding a correct inf bound is hard $\rightarrow$ no $\ell_p$ bound}

\medskip

We now derive lower bounds on the energy consumption, that will be
useful to design an approximation algorithm in the next section. 

% Usually when approximation algorithms are due, one tries to compare to
% theoretical lower bounds. In this section we try to show what possible
% bounds can be used for a theoretical comparison of any algorithm to
% the optimal one. We then give some hints why these bounds cannot be
% effectively used, then we give some Propositions that enable us to
% refine these bounds. Note that improving these bounds can only result
% in improving the approximation algorithm.

\begin{proposition}[Lower bound without reliability]
 \label{bound.lp}
%Minimizing the $\ell_2$ norm is a lower bound for the \indep problem. In 
%particular, 
 The optimal solution of \indep cannot have an energy lower than
 $\frac{S^3}{(pD)^2}$.
\end{proposition}

\begin{proof}
%As shown by \cite{RenaudGoudGreedy}, minimizing the $\ell_2$ norm is equivalent 
%to considering the energy problem with a fixed deadline. That is the same 
%problem without the reliability constraint. One less constraint implies that 
%the optimal solution is necessarily better.
Let us consider the problem of minimizing the energy consumption, with
a deadline constraint~$D$, but without accounting for the constraint
on reliability. A lower bound is obtained if the load on each
processor is exactly equal to~$\frac{S}{p}$, and the speed of each
processor is constant and equal to~$\frac{S}{pD}$. The corresponding
energy consumption is $S \times \left(\frac{S}{pD}\right)^2$, hence
the bound. 
%The particular case where we can map a load exactly equal to $\frac{S}{p}$ on 
%each processor and run it at speed $\frac{S}{pD}$ is the best case and is then 
%a lower bound.
\end{proof}

% This is the most commonly bound used when scheduling independent
% tasks with energy, but without the reliability constraint
% (see~\cite{Alon97,RenaudGoudGreedy}), and is a direct consequence of
% considering the $\ell_2$ norm.
However, if the speed $\frac{S}{pD}$ is small compared to~\fr,
the bound is very optimistic since reliability constraints are not
matched at all. Indeed, replication must be used in such a case. 
We investigate bounds that account for replication in the following,
using the optimal solution of the \chain problem. 

%  (understand,
% before \fr), this bound cannot be efficiently used since we need to
% replicate tasks in order to match the reliability constraints, and the
% $\ell_2$ norm is not anymore usable.  This calls for better bounds
% when $\frac{S}{pD}$ is ``small''. %We need to find a
% %new norm that would take into account the re-executed tasks.

% A second example of a bound is a bound where some re-execution is taken into 
% account. For this we use the optimal solution of the \chain problem.
% We first state this simple Proposition in order to improve the bound.
% \begin{proposition}
%     \label{prop.sizereex}
% For the problem \indep, if a task $T_i$ is replicated or re-executed, then 
% necessarily, $w_i<\frac{1}{\sqrt{2}}D\fr$.
% \end{proposition}
% \begin{proof}
% Since at least one of the execution of the task should be at a speed lower than 
% $\frac{1}{\sqrt{2}}\fr$ (otherwise we would not have a better energy than one 
% execution at \fr), and since the deadline is $D$, in order to match the deadline 
% the weight of the task has to be strictly lower than $\frac{1}{\sqrt{2}}D\fr$.
% \end{proof}

\begin{proposition}[Lower bound using linear chains]
    \label{bound.chain.improved}
    For the \indep problem, the optimal solution cannot have an energy
    lower than the optimal solution to the \chain problem on a single
    processor with a deadline~$pD$, where the weight of the
    re-executed tasks is lower than $\frac{1}{\sqrt{2}}D\fr$.
\end{proposition}

\begin{proof}
  We can transform any solution to the \indep problem into a solution
  to the \chain problem with deadline~$pD$ and a single
  processor. Tasks are arbitrarily ordered as a linear chain, and the
  solution uses the same number of executions and the same speed(s)
  for each task. It is easy to see that the \indep problem is more
  constrained, since the deadline on each processor must be enforced. 
  The constraint on the weights of the re-executed tasks comes from
  Proposition~\ref{prop.indep}. Therefore, the solution to the \chain
  problem is a lower bound for \indep. 
%   There is an obvious transformation from any solution to \indep to a
%   solution of \chain with deadline $pD$ and a single processor: it
%   suffices to execute the tasks as they are executed (same number of
%   executions, same speed(s)) on the $p$
%   processors. Proposition~\ref{prop.indep} allows us to add the
%   supplementary constraint on the size of the re-executed tasks.
\end{proof}

The optimal solution may however be far from this bound, since we do
not know if the tasks that are re-executed on a chain with a long
deadline~$pD$ can be executed at the same speed when the deadline
is~$D$. The constraint on the weight of the re-executed tasks allows
us to improve slightly the bound, and this lower bound is the basis of
the approximation algorithm that we design for \indep. 

% However it is very easy to see the limit of such a bound: we do not
% know if the tasks that are re-executed on a chain with a long deadline
% ($pD$) can be executed at the same speed on a smaller deadline
% $D$. The minimal energy of \chain with a deadline $pD$ can be as far
% as we want from the one of \indep with deadline $D$.  This is probably
% the reason why until now, to the best of our knowledge, no-one ever
% used this type of bound to prove approximation ratios on independent
% tasks.  However since we do not have an obvious norm, this seems to us
% the easiest bound to use.

\subsection{Approximation algorithm for \indep}
\label{algo.indep}

%The important result here is that because of the constraint due to
%\fmax and the deadline, there cannot be any approximability result
%when scheduling is due.  Indeed, even without energy, this remains a
%scheduling complete (NP-complete).

We have seen in Section~\ref{sec.inapprox} that there exists no
constant factor approximation algorithm for \indep, unless P=NP, even
without accounting for the reliability constraint. This is due to the
constraint on the makespan and the maximum speed~\fmax. Therefore, in
order to provide a constant factor approximation algorithm, we relax
the constraint on the makespan and propose an
$(\alpha,\beta)$-approxi\-mation algorithm. The solution~$E_{algo}$ is
such that $E_{algo} \leq \alpha \times E_{opt}$, where $E_{opt}$ is
the optimal solution with the deadline constraint~$D$, and the
makespan of the algorithm $M_{algo}$ is such that $M_{algo}\leq \beta
\times D$. 

The result of Section~\ref{sec.inapprox} means that for all $\alpha
>1$, there is no $(\alpha,1)$-approxi\-mation algorithm for \indep,
unless $P=NP$.  Therefore, %we aim at finding a
%$(1+\varepsilon,2-\varepsilon')$-approxi\-mation algorithm, where
%$\varepsilon$ is small: by relaxing the makespan constraint of at most
%a factor of~$2$, we can approach the optimal solution for energy
%consumption. Let $\beta$ be the relaxation on the deadline. 
we present 
%first 
an algorithm that realizes a %rather a 
$(1+\frac{1}{\beta^2},\beta)$-approximation, where the minimum relaxation on
the deadline is smaller than~$2$. It is of course
possible to run the algorithm with larger values of~$\beta$, leading
to a better guarantee on the energy consumption. 
%and then we refine it to
%have $(1+\Theta(\frac{1}{p}),\beta)$ for large values of~$p$ (number of
%processors).

%\subsubsection{$(1+\frac{1}{\beta^2},\beta)$-approximation algorithm}

\medskip
{\bf Sketch of the algorithm.}  In the first step of the algorithm, we
schedule each task with a big weight alone on one
processor, with no replication. A task~$T_i$ is considered as {\em big}
if $w_i \geq \max(\frac{S}{p},D\fr)$. This step is done
in polynomial time: we sort the tasks by non-increasing
weights, and then we check whether the current task is such that $w_i \geq
\max(\frac{S}{p},D\fr)$. If it is the case, we schedule the task alone
on a processor and we let $S=S-w_i$ and $p=p-1$. The procedure ends
when the current task is small enough, i.e., all remaining tasks
are such that $w_i < \max(\frac{S}{p},D\fr)$, with the updated values
of $S$ and~$p$. 
%
%Next, depending on the value of $\frac{S}{p}$, we explain the
%different steps of the algorithm. 

\newcounter{nbProcs}
\newcounter{fra}
\newlength{\wiwi}
\newcounter{vis}
\newcounter{ssecbu}

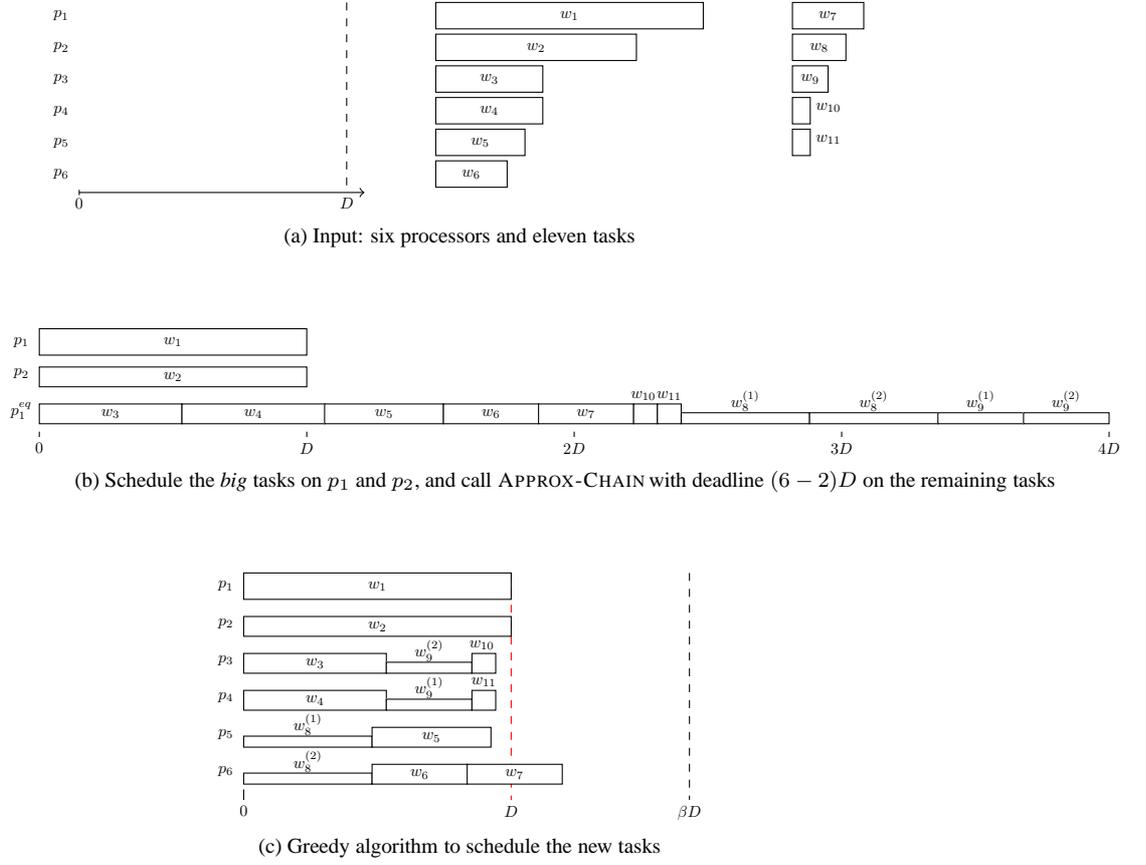
\begin{figure*}
\centering
\subfloat[Input: six processors and eleven tasks]
{\begin{tikzpicture}
\begin{scope}[xscale=0.67]

\FPeval{\procspace}{.2}
\FPeval{\cury}{0}
\FPeval{\curx}{0}

\coordinate (OFF) at (0,0);
\begin{scope}[scale=.35, every node/.style={scale=0.6},shift=(OFF)]

\inSched
\addProc
\draw[->] (0,0-36*\procspace) -- (16,0-36*\procspace);
\draw[dashed] (15,0) -- (15,0-36*\procspace) ;
\draw[draw=none] (15,0-36*\procspace) node[below] {\ensuremath{D}};
\draw (0,0-35.5*\procspace) -- (0,0-36.5*\procspace) ;
\draw[draw=none] (0,0-36*\procspace) node[below] {\ensuremath{0}};

\addProc
\addProc
\addProc
\addProc
\addProc
%\addProc

\end{scope}

 \coordinate (SH) at (7,0);
 
 \begin{scope}[scale=.35, every node/.style={scale=0.6},shift=(SH)]

\addLisJob{15}{1}{1}
\addLisJob{11.25}{1}{2}
\addLisJob{6}{1}{3}
\addLisJob{6}{1}{4}
\addLisJob{5}{1}{5}
\addLisJob{4}{1}{6}

\end{scope}
 \coordinate (SH) at (14,0);

\begin{scope}[scale=.35, every node/.style={scale=0.6},shift=(SH)]

\addLisJob{4}{1}{7}
\addLisJob{3}{1}{8}
\addLisJob{2}{1}{9}
\addLisJob[below right]{1}{1}{{10}}
\addLisJob[below right]{1}{1}{{11}}

\end{scope}
\end{scope}
\end{tikzpicture}
%\caption{Input}
%\end{figure}

}\\[1cm]
%%%%%%%%%%%%%%%%%%%%%%%%%%%%%%%%%%%%%%%%%%%%%%%%%
%%%%%%%%%%%%%%%%%%%%%%%%%%%%%%%%%%%%%%%%%%%%%%%%%
\subfloat[Schedule the {\em big} tasks on $p_1$ and $p_2$, and call
\approxchain\xspace with deadline $(6-2)D$ on the remaining tasks]
{
%\begin{figure}
%\centering
\begin{tikzpicture}
\begin{scope}[xscale=0.67]

\FPeval{\procspace}{.2}
\FPeval{\cury}{0}
\FPeval{\curx}{0}

\coordinate (OFF) at (0,0);
\begin{scope}[scale=.35, every node/.style={scale=0.6},shift=(OFF)]

\inSched
\addProc

\addJobSpeed[anchor=center]{15}{1}{1}

\addProc
\FPdiv{\speed}{3}{4}
\addJobSpeed[anchor=center]{11.25}{\speed}{2}

\addProcChain

\addJobSpeed{6}{\speed}{3}
\addJobSpeed{6}{\speed}{4}
\addJobSpeed{5}{\speed}{5}
\addJobSpeed{4}{\speed}{6}
\addJobSpeed{4}{\speed}{7}
\addJobSpeed[above=0.1]{1}{\speed}{{10}~}
\addJobSpeed[above=0.1]{1}{\speed}{{11}}
\FPdiv{\speed}{5}{12}
\addJobSpeed[above=0.2pt]{3}{\speed}{\ensuremath{8^{(1)}}}
\addJobSpeed[above=0.2pt]{3}{\speed}{\ensuremath{8^{(2)}}}
\addJobSpeed[above=0.2pt]{2}{\speed}{\ensuremath{9^{(1)}}}
\addJobSpeed[above=0.2pt]{2}{\speed}{\ensuremath{9^{(2)}}}
\FPsub{\cury}{\cury}{1}

\foreach \y in {1.5}
	{
	\draw[draw=none] (0,\cury-\y*\procspace-\procspace) node[below] {\ensuremath{0}};
	\draw (0,\cury-\y*\procspace) -- (0,\cury-\y*\procspace-\procspace) ;
	\draw[draw=none] (15,\cury-\y*\procspace-\procspace) node[below] {\ensuremath{D}};
	\draw (15,\cury-\y*\procspace) -- (15,\cury-\y*\procspace-\procspace) ;
	\foreach \x in {2,...,4}
		{
		\draw[draw=none] (15*\x,\cury-\y*\procspace-\procspace) node[below] {\ensuremath{\x D}};
		\draw (15*\x,\cury-\y*\procspace) -- (15*\x,\cury-\y*\procspace-\procspace) ;
		};
	}

\end{scope}

\end{scope}
\end{tikzpicture}
%\caption{1. Schedule the ``big'' tasks; 2. FPTAS on chain with deadline $4D$.}
%\end{figure}

}\\[1cm]
%%%%%%%%%%%%%%%%%%%%%%%%%%%%%%%%%%%%%%%%%%%%%%%%%
%%%%%%%%%%%%%%%%%%%%%%%%%%%%%%%%%%%%%%%%%%%%%%%%%
\subfloat[Greedy algorithm to schedule the new tasks]
{
%\begin{figure}
%\centering
\begin{tikzpicture}
\begin{scope}[xscale=0.67]

\FPeval{\procspace}{.4}
\FPeval{\cury}{0}
\FPeval{\curx}{0}

\coordinate (OFF) at (0,0);
\begin{scope}[scale=.35, every node/.style={scale=0.6},shift=(OFF)]

\inSched
\addProc

\draw (0,0-20.5*\procspace) -- (0,0-21.5*\procspace) ;
\draw[draw=none] (0,0-21.5*\procspace) node[below] {\ensuremath{0}};
\draw[dashed,red] (15,0) -- (15,0-21.5*\procspace) ;
\draw[draw=none] (15,0-21.5*\procspace) node[below] {\ensuremath{D}};
\draw[dashed] (25,0) -- (25,0-21.5*\procspace) ;
\draw[draw=none] (25,0-21.5*\procspace) node[below] {\ensuremath{\beta D}};

\addJobSpeed{15}{1}{1}
\addProc
\FPdiv{\speed}{3}{4}
\addJobSpeed{11.25}{\speed}{2}

\addProc
\FPdiv{\speed}{3}{4}
\addJobSpeed{6}{\speed}{3}
\FPdiv{\speed}{5}{12}
\addJobSpeed[above=0.2pt]{2}{\speed}{\ensuremath{9^{(2)}}}
\FPdiv{\speed}{3}{4}
\addJobSpeed[above=0.1]{1}{\speed}{{10}~}

\addProc
\FPdiv{\speed}{3}{4}
\addJobSpeed{6}{\speed}{4}
\FPdiv{\speed}{5}{12}
\addJobSpeed[above=0.2pt]{2}{\speed}{\ensuremath{9^{(1)}}}
\FPdiv{\speed}{3}{4}
\addJobSpeed[above=0.1]{1}{\speed}{{11}}

\addProc
\FPdiv{\speed}{5}{12}
\addJobSpeed[above=0.2pt]{3}{\speed}{\ensuremath{8^{(1)}}}
\FPdiv{\speed}{3}{4}
\addJobSpeed{5}{\speed}{5}

\addProc
\FPdiv{\speed}{5}{12}
\addJobSpeed[above=0.2pt]{3}{\speed}{\ensuremath{8^{(2)}}}

\FPdiv{\speed}{3}{4}
\addJobSpeed{4}{\speed}{6}
\addJobSpeed{4}{\speed}{7}

\end{scope}

\end{scope}
\end{tikzpicture}
}
\caption{$\left(1+\frac{1}{\beta^2},\beta \right)$-approximation algorithm for independent tasks\label{fig.indep}}
\end{figure*}

\begin{compactitem}
%\item If $\frac{S}{p} > D \fmax$, then there is no solution with a
%  deadline~$D$, since executing each task once at maximum speed is not
%  sufficient to match the deadline, even if a perfect load balancing
%  is possible. 

\item If %$D \fmax \geq
  ${S} > pD \fr $, i.e., %\frac{1}{1 + \varepsilon^*}$, i.e.,
  the load is {\em large enough}, we do not use replication, but we
  schedule the tasks at
  speed~$\frac{S}{pD}$, % $\max(\fr,\frac{S}{pD})$,
  using a simple scheduling heuristic, \dff~\cite{Graham69}. Tasks
  are sorted by non increasing weights, and at each time step, we
  schedule the current task on the least loaded processor. Thanks to
  the lower bound of Proposition~\ref{bound.lp}, the energy
  consumption is not greater than the optimal energy consumption, and
  we determine $\beta$ such that the deadline is
  enforced.  %allows us to derive a good approximation
%  ratio. The value of~$\varepsilon^*$, i.e., the threshold on the
%  total load, will be determined in the proof of the approximation
%  ratio. 

\item If $S\leq p D \fr$, %(1 + \varepsilon^*)\frac{S}{p}$,
  the previous bound is not good enough, and therefore we use the
  FPTAS on a linear chain of tasks with deadline~$pD$ for \chain (see
  Theorem~\ref{approx.chain}).  The FPTAS is called with
  \begin{equation}
\label{def.epsilon}
\varepsilon=\min\left(\frac{2w_{min}}{3S}\left(\frac{\fmin}{\fr}\right)^2,\;
    \frac{1}{3\beta^2}\right), 
\end{equation}
  where $w_{min} = \min_{1\leq i \leq n} w_i$. Note that it is
  slightly modified so that only tasks of weight
  $w<\frac{1}{\sqrt{2}}D\fr$ can be replicated, and that we enforce a
  minimum speed~$\fmin$.  The FPTAS therefore determines which tasks
  should be executed twice, and it fixes all execution speeds.
   
  We then use \dff in order to map the tasks onto the $p$ processors,
  at the speeds determined earlier. The new set of tasks includes both
  executions in case of replication, and tasks are sorted by non
  increasing execution times (since all speeds are fixed). At each
  time step, we schedule the current task on the least loaded
  processor. If some tasks cannot fit in one processor within the
  deadline~$\beta D$, we re-execute them at speed~$\frac{w_i}{\beta
    D}$ on two processors. Thanks to the lower bound of
  Proposition~\ref{bound.chain.improved}, we can bound the energy
  consumption in this case.
\end{compactitem}

\medskip
\noindent We illustrate the algorithm on an example in
Figure~\ref{fig.indep}, where eleven tasks must be mapped on six
processors. For each task, we represent its execution speed as its
height, and its execution time as its width. There are two {\em big}
tasks, of weights $w_1$ and~$w_2$, that are each mapped on a distinct
processor. Then, we have $p=4$ and we call \approxchain\xspace with
deadline~$4D$; tasks $T_8$ and $T_9$ are replicated. Finally, \dff
greedily maps all instances of the tasks, slightly exceeding the
original bound~$D$, but all tasks fit within the extended deadline.

\medskip
This algorithm leads to the following theorem: 

\begin{theorem}
    \label{thm.indep}
For the problem \indep, there are
$\left(1+\frac{1}{\beta^2},\beta \right)$-approxi\-mation
%$\left(1+\Theta(\frac{1}{p}),2-\Theta(\frac{1}{p})\right)$-approxi\-mation
algorithms, for all $\beta \geq 2-\Theta(\frac{1}{p})$, that run in
polynomial time.
%\begin{itemize}
%\item if $p \leq 3$, there is a $\left(1+\varepsilon,
%    2-\frac{3}{2p+1}\right)$-approxi\-mation algorithm that runs in time
%  polynomial in the problem size and in~$\frac{1}{\varepsilon}$;
%\item otherwise, there is a
%  $\left(1+\Theta(\frac{1}{p}),2-\Theta(\frac{1}{p})\right)$-approxi\-mation
%  algorithm that runs in polynomial time.
%\end{itemize}
%More precisely, when $p \geq 4$, we have a $\left(\left
%    (1-\frac{20}{\sqrt{2}cp} \right
%  )^{-2},2-\frac{3}{2p+1}\right)$-approxi\-mation algorithm, where $c$
%is the constant introduced in Theorem~\ref{thm.chain}.
\end{theorem}

% \begin{definition}[\dff]
%   Consider the tasks sorted by non increasing weight; at each step,
%   schedule the current task on the least-loaded processor.
% \end{definition}

Before proving Theorem~\ref{thm.indep}, we give some preliminary
results: we prove below the optimality of the first step of the
algorithm, i.e., the optimal solution would schedule tasks of weight
greater than $\max(\frac{S}{p},D\fr)$ alone on a processor:

\begin{proposition}
  \label{prop.grobt}
  In any optimal solution to \indep, %if $\frac{S}{pD} \geq \fr$, then
  each task~$T_i$ such that $w_i \geq \max(\frac{S}{p},D\fr)$ 
  % $\frac{S}{p}$, 
  is executed only once, 
 %at a speed greater or equal to $\max(\frac{S}{pD},\fr)$, 
  and it is alone on its processor.
%$T_i$ is the only task executed on its processor.
\end{proposition}

\begin{proof}
  Let us prove the result by contradiction.  Suppose that there exists
  a task~$T_i$ such that $w_i \geq \max(\frac{S}{p},\ab D\fr)$, and that
  this task is executed on processor~$p_1$.  Suppose also that there
  is another task~$T_j$ executed on~$p_1$, with $w_j \leq w_i$.
  Necessarily, there exists a processor, say~$p_2$, whose load is
  smaller than $\frac{S}{p}$, since the load of~$p_1$ is strictly
  greater than~$\frac{S}{p}$.
 % (without loss of generality we call it $p_2$).  
  Consider the energy of the tasks executed on processors $p_1$
  and~$p_2$. Because of the convexity of the energy function, it is
  strictly better to execute task~$T_j$ on processor~$p_2$, and then
  $T_i$~is executed alone on processor~$p_1$, at a speed
  $\frac{w_i}{D} \geq \fr$.  % \max(\frac{S}{pD},\fr)$. 
\end{proof}

\begin{comment}
\begin{proposition}
    \label{prop.sizereex}
For the problem \indep, if a task $T_i$ is replicated or re-executed, then 
necessarily, $w_i<\frac{1}{\sqrt{2}}D\fr$.
\end{proposition}
\begin{proof}
  Since at least one of the execution of the task should be at a speed
  lower than $\frac{1}{\sqrt{2}}\fr$ (otherwise we would not have a
  better energy than one execution at \fr), and since the deadline is
  $D$, in order to match the deadline the weight of the task has to be
  strictly lower than $\frac{1}{\sqrt{2}}D\fr$. 
\end{proof}
\end{comment}

% The next proposition is here to delimit the moment when a simple greedy 
% heuristic (here \dff, see~\cite{Graham69}) will not suffice for the energy 
% consumption. Obviously, a better scheduling algorithm for independent tasks 
% ($(P||C_{\max}) \in \text{FPTAS}$~\cite{Ausiello99}) may improve the deadline 
% relaxation. However in order not to over-complicate the proof, we chose to 
% restrict to a simple list-scheduling heuristic.

Next, we prove a lemma that will allow us to tackle the case where the
load is {\em large enough} ($S>pD\fr$), and we obtain a minimum on the
approximation ratio of the deadline~$\beta$. 

\begin{lemma}
\label{lemma.bigload}
%\begin{proposition}
% \label{prop.epsilon1}
 For the problem \indep where each task~$T_i$ is such that
 $w_i<\max(\frac{S}{p},D\fr)$,  %let $\varepsilon \geq 0$, and suppose
% that $(1+\varepsilon)\frac{S}{pD} > \fr$. Then, 
 scheduling each task only once at speed $\max(\fr,\frac{S}{pD})$ with
 the \dff heuristic leads to a make\-span of at most $\beta D$, with
 $\beta \!=\! \max\left(\!2\!-\!\frac{3}{2p+1}, 2\!-\!\frac{p+2}{4p+2}\!\right)$.
\end{lemma}

Note that we introduce $\max(\frac{S}{p},D\fr)$ since the lemma is
also used in the case $S\leq pD\fr$. Also, since $\beta$ is
increasing with~$p$ and the bound is computed in fact for a number of
processors smaller than the original one (some processors are
dedicated to {\em big} tasks), the value of $\beta$ computed with the
total number of processors~$p$ is not smaller and it is possible to
achieve a makespan of at most~$\beta D$. 

\begin{proof}
Let $l_{\text{dff}}$ be the maximal load of the processors after
applying \dff on the weights of the tasks.  Let us find $\beta$ such
that $l_{\text{dff}}\frac{pD}{S} \leq \beta D$: this means that within
a time $\beta D$, we can schedule all tasks %each task on the processor returned
%by \dff 
at speed $\frac{S}{pD}$, and therefore at speed $\max(\fr, \ab
\frac{S}{pD})$, since the most loaded processor succeeds to be within
the deadline $\beta D$.

Let $l_{\text{opt}}$ be the maximal load of the processors in an
optimal solution, and let $T_i$ be the last task executed on the
processor with the maximal load $l_{\text{dff}}$ by \dff. We have
either $w_i\leq l_{\text{opt}}/3$ or $w_i> l_{\text{opt}}/3$.

\paragraph{$\bullet$} If $w_i \leq l_{\text{opt}}/3$, we know that
$l_{\text{opt}} \!\leq\! l_{\text{dff}} \!\leq\! \left(\! \frac{4}{3} \!-\!
  \frac{1}{3p}\! \right)l_{\text{opt}}$, since \dff is a $\left(
  \frac{4}{3} - \frac{1}{3p} \right)$-approxi\-mation
\cite{Graham69}. 
%\fw{What could we say (lower bound) using the FPTAS for indep task scheduling?}
We want to compare $l_{\text{opt}}$ to $S / p$ (average load).  We
consider the solution of \dff. At the time when $T_i$ was scheduled,
all the processors were at least as loaded as the one on which $T_i$
was scheduled, and hence we obtain a lower bound on~$S$: $S \geq
(p-1)(l_{\text{dff}}-w_i) + l_{\text{dff}}$.  Furthermore,
$l_{\text{dff}}-w_i \geq \frac{2}{3} l_{\text{opt}}$ (because
$l_{\text{dff}}\geq l_{\text{opt}}$ and $w_i \leq l_{\text{opt}}/3$).
%Since $l_{\text{opt}} \leq l_{\text{dff}}$, we have $w_i \geq l_{\text{opt}} - (l_{\text{dff}} -w_i)$.  That is:
Finally, $S \geq (p-1)\frac23 l_{\text{opt}} + l_{\text{opt}}$, which
means that $l_{\text{opt}} \leq \frac{S}{p} \frac{3p}{2p+1}$, and
$l_{\text{dff}} \leq \left ( \frac{4}{3} - \frac{1}{3p} \right
)\frac{3p}{2p+1}\frac{S}{p} =  \left(2-\frac{3}{2p+1}\right)\frac{S}{p}$.

In this case, with $\beta = 2-\frac{3}{2p+1}$, we can execute all the
tasks at speed
%We have this first result: 
%if $w_i \leq l_{\text{opt}}/3$ and $\beta \geq \left ( \frac{4}{3} - \frac{1}{3p} \right )\frac{3p}{2p+1} = 2-\frac{3}{2p+1}$, 
%then we can execute all the tasks on the processor of maximal load at speed 
$\max(\fr,\frac{S}{pD})$ within the deadline~$\beta D$.

\paragraph{$\bullet$} 
If $w_i > l_{\text{opt}}/3$, %we can assume without loss of generality
%that the tasks are sorted by decreasing weight.
it is known that \dff is optimal for the execution time
\cite{Graham69}, i.e., $l_{\text{opt}} = l_{\text{dff}}$, and we aim
at finding an upper bound on~$l_{\text{opt}}$.  We assume in the
following that tasks are sorted by non increasing weights.

If $w_i \geq \frac{S}{p}$, then we show that % $w_i < D\fr$ and
$T_i$~is the only task executed on its processor (recall that $T_i$~is
the last task executed on the processor with the maximal load by
\dff). Indeed,
%\begin{compactitem}
%\item 
there cannot be $p$ tasks of weight not smaller than~$\frac{S}{p}$,
hence $i<p$, and $T_i$ is the first task scheduled on its
processor. Moreover, 
%\item 
if \dff were to schedule another task on the processor of~$T_i$,
  then this would mean that the $p-1$ other processors all have a load
  greater than~$w_i$, and hence the total load would be 
  greater than~$S$. 
%\end{compactitem}
%
Then, since $w_i< \max(\frac{S}{p},D\fr)$ and $w_i\geq \frac{S}{p}$,
we have $w_i < D\fr$ and we can execute each task at speed $\fr =
\max(\fr,\frac{S}{pD})$ within a deadline~$D$. Indeed, the maximal
load is then~$w_i$, by definition of~$T_i$. Therefore, the result
holds (with $\beta=1$).

Now suppose that $w_i < \frac{S}{p}$. In that case, if $T_i$ was the
only task executed on its processor, then we would have
$l_{\text{opt}} = l_{\text{dff}} < \frac{S}{p}$, which is impossible
since $S = \sum_{k=1}^p l_k \leq p l_{\text{opt}}$.  Therefore, $T_i$
is not the only task executed on its processor.  A direct consequence
of this fact is that $p+1 \leq i$. Indeed, \dff schedules the $p$
largest tasks on $p$ distinct processors; since $T_i$ is the last
task scheduled on its processor, but not the only one, then $T_i$ is
not among the $p$ first scheduled tasks. Also, there are only two
tasks on the processor executing~$T_i$, since $w_i> l_{\text{opt}}/3$
and the tasks scheduled before~$T_i$ have a weight at least equal
to~$w_i$. Finally, $p+1\leq i \leq 2p$. 

After scheduling task $T_j$ on processor $j$ for $1 \leq j \leq p$,
\dff schedules task $T_{p+j}$ on processor $p-j+1$ for $1 \leq j \leq
i-p$, and $T_i$~is therefore scheduled on processor~$p_{2p-i+1}$,
together with task $T_{2p-i+1}$, and we have   
$w_i+w_{2p-i+1} = l_{\text{opt}}$.
Note that because the $w_j$ are sorted, $S\geq \sum_{j\leq i}w_j \geq i w_i$. 
We also have $w_{2p-i+1} < \frac{S}{p}$: indeed, when $T_i$ was
scheduled, the load of the $p$ processors was at least equal to the
load of the processor where $T_{2p+i-1}$ was scheduled. Hence, 
$w_{2p-i+1}$ cannot be greater than~$\frac{S}{p}$.  Then, since
$w_{2p-i+1} =l_{\text{opt}}-w_i$,  $w_i > l_{\text{opt}} -
\frac{S}{p}$, and finally %$w_i \in \left]l_{\text{opt}} -
%\frac{S}{p},\frac{S}{i}\right]$
$l_{\text{opt}} - \frac{S}{p} < w_i \leq \frac{S}{i}$. 

In order to find an upper bound on~$l_{\text{opt}}$, we provide a
lower bound to~$S$, as a function of~$w_i$: 
\begin{align*}
S & = \sum_{j=1}^n w_j \geq \sum_{j=1}^i w_j = \sum_{j=1}^{2p-i+1} w_j + \sum_{j=2p-i+2}^i w_j \\
 &\geq (2p-i+1) w_{2p-i+1} + (2(i-p)-1)w_i \\
 & = (2p-i+1) (l_{\text{opt}}-w_i) + (2(i-p)-1)w_i \\
 &= (2p-i+1) l_{\text{opt}} + (3i-4p-2) w_i = f(w_i). 
\end{align*}

We then have $f'(w_i) = 3i-4p-2$, and we consider two cases. 
%By differentiating $f$ with respect to~$w_i$, knowing that $w_i \in ]l_{\text{opt}} - \frac{S}{p},\frac{S}{i}[$, we can separate two cases
%($f'(w_i) = 3i-4p+2$).
%\medskip
% $\bullet$ 

If $f'(w_i) \geq 0$, 
then we have $i\geq \frac{4p+2}{3}$, and finally 
$S \geq i w_i > \frac{4p+2}{3} \left(l_{\text{opt}} - \frac{S}{p}\right).$
We can conclude that $l_{\text{opt}} < \frac{S}{p} \left( 1+ \frac{3p}{4p+2}
\right) = \frac{S}{p} \left( 2-\frac{p+2}{4p+2} \right)$. 
%That is:
%\[ \frac{S}{p} \left( 1+ \frac{3p}{4p+2} \right) \geq l_{\text{opt}}.\]
%Hence this second result: 
%if $w_i < l_{\text{opt}}/3$, $3i-4p+2 \geq 0$ and $\beta \geq 2-\frac{p+2}{4p+2}$, 
%then we can execute all the tasks on the processor of maximal load at speed 
%$\max(\fr,\frac{S}{pD})$.

%\medskip
% $\bullet$  
%If $3i-4p+2 < 0$:
Otherwise, $f'(w_i)<0$ and $f$~is a decreasing function of~$w_i$,
i.e., its minimum is reached when $w_i$ is maximal, and $S\geq
f(\frac{S}{i})$.
 Hence, $S \geq (2p-i+1) l_{\text{opt}} + (3i-4p-2) \frac{S}{i}$. 
Since $i\leq 2p$, $2p-i+1>0$ and 
$$l_{\text{opt}} \leq \frac{S}{i} \left( \frac{i-3i+4p+2}{2p-i+1}
\right) = \frac{2S}{i}. $$
Finally, since $i\geq p+1$, $l_{\text{opt}} \leq
\frac{2S}{p+1} = \frac{S}{p}\left(2-\frac{2}{p+1}\right)$. 

% We showed that 
% $2p \geq i\geq p+1$.
% \begin{align*}
% S  &\geq (2p-i+1) l_{\text{opt}} + (3i-4p-2) \frac{S}{i} \\
% (\frac{i}{i} + \frac{4p+2-3i}{i}) S &\geq (2p-i+1)  \\
% \frac{2}{i} S & \geq l_{\text{opt}} \qquad \text{\emph{Because $2p-i+1>0$.}}\\
% \frac{2p}{p+1} \frac{S}{p}  & \geq l_{\text{opt}} \qquad \text{\emph{Because $i\geq p+1$.}}\\
% (2-\frac{2}{p+1})\frac{S}{p} & \geq l_{\text{opt}}
% \end{align*}

% Hence this last result: 
% if $w_i < l_{\text{opt}}/3$, $3i-4p+2 < 0$ and $\beta \geq 2-\frac{2}{p+1}$, 
% then we can execute all the tasks on the processor of maximal load at speed 
% $\max(\fr,\frac{S}{pD})$. Hence the proposition.

Overall, if $w_i > l_{\text{opt}}/3$, we have the
bound $$l_{\text{opt}} \leq \frac{S}{p} \times \max\left( 2 -
  \frac{p+2}{4p+2}, 2-\frac{2}{p+1} \right).$$ Therefore, for $ \beta
\geq \max\left( 2 - \frac{p+2}{4p+2}, 2-\frac{2}{p+1} \right)$, we can
execute all the tasks on the processor of maximal load (and hence all
the tasks) at speed $\max(\fr,\frac{S}{pD})$ within the deadline
$\beta D$ in the case $w_i > l_{\text{opt}}/3$. 

\medskip We can now conclude the proof of
Lemma~\ref{lemma.bigload} by saying that for $\beta =
\max\left(2\!-\!\frac{3}{2p+1},2\!-\!\frac{p+2}{4p+2},
  2\!-\!\frac{2}{p+1}\right)$, i.e., $\beta =
\max\left(2-\frac{3}{2p+1},2-\frac{p+2}{4p+2}\right)$, %\dff is a
%$\left( (1+\varepsilon)^2, \beta\right)$-approximation algorithm,
scheduling each task only once at speed $\max(\fr,\frac{S}{pD})$ with
 the \dff heuristic leads to a make\-span of at most $\beta D$.
 \end{proof}

We are now ready to prove Theorem~\ref{thm.indep}. 

\begin{proof}[{\bf Proof of Theorem~\ref{thm.indep}}]
First, thanks to Proposition~\ref{prop.grobt}, we know that the first
step of the algorithm takes decisions that are identical to the
optimal solution, and therefore these tasks that are executed once,
alone on their processor, have the same energy consumption than the
optimal solution and the same deadline. We can therefore safely ignore
them in the remaining of the proof, and consider that for each task~$T_i$, 
$w_i<\max(\frac{S}{p},D\fr)$. 

%Recall that $\varepsilon^*$ is the threshold value on the load used in
%the approximation algorithm; we have $\varepsilon^*\geq 0$, and its
%exact value will be determined later in the proof.  
%
\medskip
In the case where %$(1 + \varepsilon^*)\frac{S}{p} > D \fr$,
$S>pD\fr$, we use the fact that $S(\frac{S}{pD})^2$ is a lower bound on the
  energy (Proposition~\ref{bound.lp}).  Each task is executed once
  at speed $\max(\fr,\frac{S}{pD})=\frac{S}{pD}$, and therefore the
  energy consumption is equal to the lower bound $S(\frac{S}{pD})^2$. 
  The bound on the deadline is obtained by applying
  Lemma~\ref{lemma.bigload}. 
%thanks to Proposition~\ref{prop.epsilon1}, the algorithm is a 
%$\left(\left( 1 + \varepsilon^* \right)^2,\beta\right)$-approximation,
%thanks to Lemma~\ref{bigload}    with $\beta= \max \left( 2-\frac{3}{2p+1}, 2
%      -\frac{p+2}{4p+2}\right)$.

% In this case we use the \dff heuristic to schedule the tasks on the $p$ 
% processors at speed $\max(\fr,\frac{S}{pD})$. We do not use re-execution nor 
% replication.
% We showed in Proposition~\ref{prop.epsilon1} that 
% if for all $T_i \in V$, $w_i<\max(\frac{S}{p},D\fr)$, then scheduling every 
% task of $V$ on the $p$ processors, at speed $\max(\fr,\frac{S}{pD})$, with the 
% \dff heuristic is a $\left( \frac{1}{1-\varepsilon_1} \right )^2$-approxi\-mation 
% for energy, and can be done within a ratio $2-\frac{3}{2p+1}$ of the deadline.

%\begin{itemize}
%    \item {\bf If $D \fr(1 - \varepsilon_1) \geq \frac{S}{p} $.}
%\end{itemize}

\medskip
We now focus on the case $S\leq pD\fr$. 
% $D \fr\geq (1 + \varepsilon^*) \frac{S}{p}$. 
Therefore, in the following, $\max(\frac{S}{pD},\fr)=\fr$.  The
algorithm runs the FPTAS on a linear chain of tasks with
deadline~$pD$, and $\varepsilon$ as defined in
Equation~\eqref{def.epsilon}. The FPTAS returns a solution on the
linear chain with an energy consumption~$E_{\text{FPTAS}}$ such that
$E_{\text{FPTAS}} \leq \left( 1 +
  \varepsilon \right)^2 E_{\text{chain}}$, where $E_{\text{chain}}$ is
the  optimal energy consumption for \chain
with deadline~$pD$ on a single processor. According to
Proposition~\ref{bound.chain.improved}, since the solution for the
linear chain is a lower bound, the optimal solution of \indep is such that 
$E_{opt}\geq E_{\text{chain}}$. 

%a solution whose energy is at a ratio $\left( 1 +
%  \varepsilon \right)^2$ of the optimal energy consumption for \chain
%with deadline~$pD$ on a single processor, denoted~$E_{\text{chain}}$.

% The FPTAS is called with
%   $\varepsilon=\frac{w_{min}}{S}\left(\frac{\fmin}{\fr}\right)^2$,
%   where $w_{min} = \min_{1\leq i \leq n} w_i$. Note that it is
%   slightly modified so that only tasks of weight
%   $w<\frac{1}{\sqrt{2}}D\fr$ can be replicated, and that we enforce a
%   minimum speed~$\fmin$.  The FPTAS therefore determines which tasks
%   should be executed twice, and it fixes all execution speeds. 
%   We then use \dff in order to map the tasks onto the $p$ processors,
%   at the speeds determined earlier. The new set of tasks includes both
%   executions in case of replication, and tasks are sorted by non
%   increasing execution times (since all speeds are fixed). At each
%   time step, we schedule the current task on the least loaded
%   processor. If some tasks cannot fit in one processor within the
%   deadline~$\beta D$, we re-execute them at speed~$\frac{w_i}{\beta
%     D}$ on two processors. 
% Thanks to the bound of Proposition~\ref{bound.chain.improved}, if we
% can schedule the solution given by the first step of this algorithm
% within the relaxed deadline, then we have a $\left(
%   1+\varepsilon^* \right)^2$-approxi\-mation to \indep.

For each task~$T_i$, let $f_i^{\text{chain}}$ be the speed of its
execution returned by the FPTAS for \chain. Note that in case of
re-execution, then both executions occur at the same speed
(Lemma~\ref{lemma.speed.chain}).  We now consider the \indep problem
with the set of tasks~$\tilde{V}$: for each task~$T_i$,
$\tilde{T_i}\in \tilde{V}$ and its weight is $\tilde{w_i} =
w_i\frac{\fr}{f_i^{\text{chain}}}$; moreover, if $T_i$ is re-executed,
we add two copies of~$\tilde{T_i}$ in~$\tilde{V}$.  Then,
$\sum_{\tilde{T_i}\in \tilde{V}} \frac{\tilde{w_i}}{\fr} = pD$ by
definition of the solution of \chain.
 % (where $\delta_i$ is the number of time $T_i$ is executed):
%indeed, by definition $f_i^{\text{chain}}$ is such that
%$\sum_{\tilde{T_i}\in \tilde{V}} \frac{w_i}{f_i^{\text{chain}}} =
%pD$.

Let $\beta= \max ( 2-\frac{3}{2p+1}, 2 -\frac{p+2}{4p+2} )$ be the
relaxation on the deadline that we have from
Lemma~\ref{lemma.bigload}. The goal is to map all the tasks
of~$\tilde{V}$ at speed~\fr within the deadline~$\beta D$, which
amounts at mapping the original tasks at the speeds assigned by the
FPTAS: %, and hence obtaining the same energy consumption than the lower
%bound of the energy consumption of the optimal solution. 
%
%\item If for all $i$, $\tilde{w_i}< D\fr$, we can apply
%  Lemma~\ref{lemma.bigload} and we have the approximation needed
%  within the deadline relaxation with \dff.
\begin{itemize}
\item If there are tasks~$\tilde{T_i}$ such that $\frac{\tilde{w_i}}{\fr} >
\beta D$, we execute them at speed~$\frac{\tilde{w_i}}{\beta D}$ alone
on their processor, so that they reach exactly the deadline~$\beta
D$. Note that in this case, the energy consumption of the algorithm
becomes greater than $E_{\text{FPTAS}}$, since we execute these tasks
faster than the FPTAS to fit on the processor. 
\item Tasks~$\tilde{T_i}$ such that $D \leq \frac{\tilde{w_i}}{\fr}
  \leq \beta D$ are executed alone on their processor at speed~\fr. 
\item For the remaining tasks and processors, we use \dff as in
  Lemma~\ref{lemma.bigload}. Since the previous tasks take a time of
  at least~$D$ in the solution of the FPTAS, and they are mapped alone
  on a processor, we can safely remove them and apply the lemma. Note
  that the number of processors may now be smaller than~$p$, hence
  leading to a smaller bound~$\beta$.
\end{itemize}

% \item If for all $i$, $\tilde{w_i} \leq \beta D\fr$, we schedule
%   tasks $\tilde{T_i}$ such that $D \leq
%   \frac{\tilde{w_i}}{\fr} \leq \beta D$ alone on their processor at
%   speed~\fr. Then for the remaining tasks and processors, we use \dff
%   as in Lemma~\ref{lemma.bigload}.  %Again, we have the approximation
% %  needed within the deadline relaxation by a simple call to \dff.

% \item If there exists~$i$ such that $\frac{\tilde{w_i}}{\fr} > \beta
%   D$, we execute them at speed~$\frac{\tilde{w_i}}{\beta D}$ alone on
%   their processor, so that they reach exactly the deadline~$\beta D$. 
% \end{enumerate}

In the end, all tasks are mapped within the deadline~$\beta D$ (where
$\beta$ is computed with the original number of processors). There
remains to check the energy consumption of the solution returned by
this algorithm. 

\medskip
If all tasks are such that $\tilde{w_i} \leq \beta D\fr$, 
$E_{algo}=E_{\text{FPTAS}} \leq \left( 1 + 
  \varepsilon \right)^2 E_{\text{chain}} \leq \left( 1 +
  \varepsilon \right)^2 E_{opt}$. \\
According to Equation~\eqref{def.epsilon}, $\varepsilon \leq
\frac{1}{3\beta^2}$, and therefore 
$$E_{algo} \leq \left(1 + \frac{2}{3\beta^2} +
  \frac{1}{9\beta^4}\right) E_{opt} \leq
\left(1+\frac{1}{\beta^2}\right)E_{opt}.$$
%hence the result. 

\medskip
Otherwise, let $\tilde{V'}$ be the set of tasks~$\tilde{T_i}$ such that
$\tilde{w_i} > \beta D\fr$. For $\tilde{T_i}\in\tilde{V'}$,  $w_i > \beta D
f_i^{\text{chain}}$. Since $w_i<D\fr$ (larger tasks have been
processed in the first step of the algorithm), we have
$f_i^{\text{chain}}<\fr$.  This means that $T_i$ belongs to the set of
the tasks that are re-executed by the FPTAS.  Hence, since we enforced
an additional constraint, we have $w_i<\frac{1}{\sqrt{2}}D\fr$. The
least energy consumed for this task by any solution to \indep is
therefore obtained when re-execu\-ting task~$T_i$ on two distinct
processors at speed~$\frac{w_i}{D}$, in order to fit within the
deadline~$D$. Task~$T_i$ appears two times in~$\tilde{V'}$, and we let
$\tilde{E}$ be the minimum energy consumption required in the optimal
solution for tasks of~$\tilde{V'}$:
$\tilde{E} = \sum_{\tilde{T_i}\in\tilde{V'}}
      w_i\left(\frac{w_i}{D}\right)^2 $.

\smallskip 
The algorithm leads to the same energy consumption as the FPTAS except
for the tasks of~$\tilde{V'}$ that are removed from the set~$X$ of
replicated tasks, and that are executed at speed~$\frac{w_i}{\beta
  D}$:
\begin{equation*}
\begin{array}{r}
E_{algo} = (S-X) \fr^2 + (2X-\sum_{\tilde{T_i}\in\tilde{V'}}
w_i)\freex^2 \\ + \sum_{\tilde{T_i}\in\tilde{V'}} w_i
\left(\frac{w_i}{\beta D} \right)^2.\end{array}
\end{equation*} 
Since $E_{\text{FPTAS}}=(S-X) \fr^2 + 2X\freex^2$, we obtain
\begin{equation*}\begin{array}{c}
E_{algo} = E_{\text{FPTAS}} + \frac{1}{\beta^2}\tilde{E} -
\sum_{\tilde{T_i}\in\tilde{V'}}w_i\freex^2. \end{array}\end{equation*}

Furthermore, $\tilde{E}\leq E_{opt}$ since it considers only the
optimal energy consumption of a subset of tasks. We have $E_{\text{FPTAS}}
\leq (1+\varepsilon)^2 E_{opt}$, and from
Proposition~\ref{prop_WC_fr}, it is easy to see that 
$E_{\text{FPTAS}}\leq S\fr^2$, i.e., $E_{\text{FPTAS}}$ is smaller than
the energy of every task executed once at speed~\fr. Hence, $E_{\text{FPTAS}}
\leq (1\!+\!\varepsilon)^2 \min(E_{opt},S\fr^2)$, 
% (the fact that $E_{\text{FPTAS}}$ is smaller than
% the energy of every task executed once at speed \fr ($S \fr^2$) is
% always true when at least a task is re-executed is a direct
% corollary of Proposition~\ref{prop_WC_fr}).
and since $\varepsilon < 1$, $(1+\varepsilon)^2 %\min(E_{opt},S\fr^2) 
\leq 1+3\varepsilon$. Finally, 
$E_{\text{FPTAS}} \leq E_{opt}+3\varepsilon S\fr^2$. 
Thanks to Equation~\eqref{def.epsilon}, 
$3\varepsilon S\fr^2 \leq 2 w_{min} \fmin^2 \leq
\sum_{\tilde{T_i}\in\tilde{V'}}w_i\freex^2$ (note that there are at
least two tasks in~$\tilde{V'}$, since tasks are duplicated). 

Finally, reporting in the expression of $E_{algo}$, %we obtain
\begin{equation*}
\begin{array}{rcl} 
E_{algo}\! &\leq& %(1+\varepsilon)^2  \min(E_{opt},S\fr^2) 
 E_{opt} \!+ 3\varepsilon S \fr^2 \!+ \frac{1}{\beta^2} E_{opt}
- \sum_{\tilde{T_i}\in\tilde{V'}}w_i\freex^2 \\
% &\leq& \left(1+\frac{1}{\beta^2}\right) E_{opt} + 3\varepsilon S\fr^2
% -\sum_{\tilde{T_i}\in\tilde{V'}}w_i\freex^2\\
&\leq& \left(1+\frac{1}{\beta^2}\right) E_{opt}.
 \end{array}
\end{equation*} 
%hence the result. 

\medskip
To conclude, we point out that this algorithm is polynomial in the
size of the input and in $\frac{1}{\varepsilon}$. 
\end{proof}

We can improve the approximation ratio on the energy for large values
of~$p$. The idea is to avoid the case in which tasks are replicated by the
chain but are not fitting within~$\beta D$ because the speed at which
they are re-executed is too small. To do so, we fix a value
$\varepsilon^*=\Theta\left(\frac{1}{p}\right)$, such that
$0<\varepsilon^*<1$ for $p\geq 24$. The variant of the algorithm is
used only when $p\geq 24$ (after scheduling the big tasks). 
The algorithm decides that the load is large enough when $S>pD\fr
\frac{1}{1+\varepsilon^*}$, leading to a
$((1+\varepsilon^*)^2,\beta)$-approximation in this case. In the other
case ($S\leq pD\fr \frac{1}{1+\varepsilon^*}$), it is possible to
prove that when there are tasks such that
$\frac{\tilde{w_i}}{\fr}>\beta D$, then necessarily all tasks are
re-executed. Next we apply Theorem~\ref{thm.chain} while fixing values
for the $\finf$'s, so as to obtain in polynomial time the optimal
solution with new execution speeds, that can all be scheduled
within~$\beta D$ using Lemma~\ref{lemma.bigload}.  Details can be
found in the appendix.

\section{Conclusion}
\label{sec.conclusion}
In this paper, we have designed efficient approximation algorithms for
the tri-criteria energy/reliability/make\-span problem, using
replication and re-execution to increase the reliability, and dynamic
voltage and frequen\-cy scaling to decrease the energy consumption.
Because of the antagonistic relation between processor speeds and
reliability, this tri-criteria problem is much more challenging than
the standard bi-criteria problem, which aims at minimizing the energy
consumption with a bound on the makespan, without accounting for a
constraint on the reliability of tasks.

We have tackled two classes of applications. For linear chains of
tasks, we propose a fully polynomial time approximation scheme.
However, we show that there exists no constant factor approximation
algorithm for independent tasks, unless P=NP, and we are able in this
case to propose an approximation algorithm with a relaxation on the
makespan constraint: with a deadline at most two times larger than the
original one, we can approach the optimal solution for energy
consumption. 

%Note that we use \dff \cite{Graham69} for simplicity; 
%however, 
As future work, it may be possible to improve the deadline relaxation
%may be improved 
by using a FPTAS to schedule independent tasks \cite{Ausiello99}
rather than \dff \cite{Graham69}.  Also, 
an open problem is to find approximation algorithms for the
tri-criteria problem with an arbitrary graph of tasks. Even though
efficient heuristics have been designed with re-execution of 
tasks (but no replication) by \cite{rr7757}, it is not clear how to
derive approximation ratios from these heuristics. It would be
interesting to 
design efficient algorithms using replication and re-execution for
the general case, and to prove approximation ratios on these
algorithms. A first step would be to tackle fork and fork-join graphs,
inspired by the study on independent tasks. 

\paragraph*{Acknowledgements:}
%The authors are with Universit\'e de Lyon, France. 
% A.~Benoit is with the Institut Universitaire de France.
This work was supported in part by the ANR {\em RESCUE} project.

\bibliographystyle{abbrv}
\bibliography{biblio}

\clearpage
\appendix

\section*{Appendix:
  $(1+\Theta(\frac{1}{p}),2-\Theta(\frac{1}{p}))$-approximation
  algorithm for \indep}

This algorithm is used only for $p\geq 24$, and we define: 
$$K = 1- \frac{1}{c(2\beta \sqrt{2} -1)};$$
$$   \varepsilon^* = \frac{1}{\sqrt{2}cpK-1} . $$

Recall that $\beta = \max(2-\frac{3}{2p+1}, 2-
\frac{p+2}{4p+2})$. The value~$\beta$ is therefore increasing
with~$p$, and for $p\geq 24$, we have $\beta \geq 1.9$.
 Furthermore, $c\approx 0.2838$ and $K\geq 0.2$. Finally, since $p\geq
 24$, $0<\varepsilon^* <1$. 

% We focus below on the case $p\geq 2$, since for $p=1$, the problem
% with independent tasks is equivalent to the linear chain problem, and
% we can use the FPTAS of Theorem~\ref{approx.chain} with $\beta=1$. 

\medskip
\noindent{\bf Modifications to the original algorithm.} \\
The handling of {\em big} tasks is identical. However, we do not use
replication when $S> pD\fr \frac{1}{1+\varepsilon^*}$: 
we schedule tasks at speed $\max(\fr,\frac{S}{pS})$ using
\dff. Proposition~\ref{prop.epsilon1} below shows that we obtain the
desired guarantee in this case.  
In the other case ($S \leq pD\fr \frac{1}{1+\varepsilon^*}$), we apply
the FPTAS with the parameter~$\varepsilon^*$. It is now possible to
show that (i) either we can schedule all tasks with the speeds returned by the
FPTAS within the deadline~$\beta D$; (ii) or there is at least one
task that does not fit, but then all tasks are re-executed and we can find
an optimal solution that can be scheduled thanks to
Theorem~\ref{thm.chain}. The correction of this case is proven in
Proposition~\ref{prop.2}. 

\begin{proposition}
 \label{prop.epsilon1}
 For the problem \indep where each task~$T_i$ is such that
 $w_i\!<\!\max(\frac{S}{p},D\fr)$, if  $(1\!+\!\varepsilon^*)\frac{S}{pD} \!>\! \fr$, then scheduling each
 task only once at speed $\max(\fr,\frac{S}{pD})$ with \dff
 is a $\left( \left( 1+\varepsilon^* \right)^2,
   \beta\right)$-approxi\-mation algorithm, with $\beta \!=\!
 \max\left(\!2\!-\!\frac{3}{2p+1}, 2\!-\!\frac{p+2}{4p+2}\!\right)$.
\end{proposition}

\begin{proof}
  We use the fact that $S(\frac{S}{pD})^2$ is a lower bound on the
  energy (Proposition~\ref{bound.lp}).  If each task is executed once
  at speed $\max(\fr,\frac{S}{pD})$, since $\fr <
  (1+\varepsilon)\frac{S}{pD}$, then the energy consumption is
  at most at a ratio $\left( 1\!+\!\varepsilon^* \right )^2$ of
  the value of the optimal energy consumption. 
  The bound on the deadline is obtained by applying
  Lemma~\ref{lemma.bigload}. 
\end{proof}

\begin{proposition}
\label{prop.2}
For the problem \indep where each task~$T_i$ is such that
 $w_i<\max(\frac{S}{p},D\fr)$, if $S \leq pD\fr
 \frac{1}{1+\varepsilon^*}$, then there is a $\left( \left(
     1+\varepsilon^* \right)^2, \beta\right)$-approxi\-mation
 algorithm, with $\beta \!=\!
 \max\left(\!2\!-\!\frac{3}{2p+1}, 2\!-\!\frac{p+2}{4p+2}\!\right)$. 
\end{proposition}

\begin{proof}
  Similarly to the original algorithm, we use the FPTAS and we obtain
  a $\left( \left( 1+\varepsilon^* \right)^2,
    \beta\right)$-approxi\-mation algorithm unless there is a task~$T_i$ such
  that $\frac{\tilde{w_i}}{\fr} > \beta D$, and hence
  $\frac{w_i}{f_i^{\text{chain}}} > \beta D$. Since $w_i<D\fr$ (larger
  tasks have been processed in the first step of the algorithm), we
  have $f_i^{\text{chain}}<\fr$.  This means that $T_i$ belongs to the
  set of the tasks that are re-executed by \approxchain.  Hence, since
  we enforced an additional constraint, we have
  $w_i<\frac{1}{\sqrt{2}}D\fr$.  Finally,
\begin{equation}
    \label{eq.fchain}
f^{\text{chain}}_i = \freex < \frac{w_i}{\beta D}< \frac{1}{\sqrt{2} \beta}\fr.
\end{equation}

%Since \freex is very small, the re-executed tasks may not fit within
%the relaxed deadline on a single processor. We analyze the cases in
%which this problem occurs. 

% According to Corollary~\ref{cor.energy.chain}, \freex is a function
% increasing with~$X$, where $X$ is the sum of the weights of the
% re-executed tasks. Hence, when \freex is ``too small'', we show that
% it can be because (i) $X_{\text{opt}}$ (returned by
% \xopt) %(Lemma~\ref{lemma.xopt})
% is very small, and in that case we can use $\varepsilon^*$ defined
% above in order to take care of this case, or (ii) $X_{\text{opt}}$ is
% not small, but then we can show that $X$ contains all tasks that can
% be re-executed, and then using Theorem~\ref{thm.chain} we can finish
% the proof by giving the optimal solution in polynomial time.

Let $X_{\text{chain}}$ be the total weight of the re-executed tasks
($X_1$ or $X_2$ in \approxchain), and let $X_{\text{opt}}=c(pD\fr -
S)$ be the optimal weight to solve \chain with one processor. 
We compute $X_{\text{opt}} - X_{\text{chain}}$.  By
definition of $\freex$ (Corollary~\ref{cor.energy.chain}), the optimal
speed at which each re-execution should occur, we have:
\[pD = \frac{S-X_{\text{chain}}}{\fr} +
\frac{2X_{\text{chain}}}{\freex} = \frac{S-X_{\text{opt}}}{\fr} +
\frac{2X_{\text{opt}}}{f_{\text{opt}}}, \]
where $f_{\text{opt}} = \frac{2c}{1+c} \fr$ (Corollary~\ref{cor.energy.chain} 
applied to $X_{\text{opt}}$). We now express $X_{\text{opt}} -
X_{\text{chain}}$: 
%\begin{align*}
$$\left ( \frac{2}{\freex} \!-\! \frac{1}{\fr} \right )\! X_{\text{chain}} =\!
\left ( 2\frac{1+c}{2c}\frac{1}{\fr} \!-\! \frac{1}{\fr} \right ) \!
X_{\text{opt}},$$
and therefore 
%\frac{2\fr - \freex}{\fr \freex} X_{\text{chain}} &= \frac{1}{c \fr} X_{\text{opt}} \\
$X_{\text{chain}} = \frac{\freex}{c(2\fr - \freex)} X_{\text{opt}}$,
and finally $X_{\text{opt}} - X_{\text{chain}} =\left (1-
  \frac{\freex}{c(2\fr -\freex)} \right) X_{\text{opt}}$, 
%\end{align*}
that is minimized when $\freex$ is maximized. Applying the upper bound
on $\freex$ from Equation~\eqref{eq.fchain}, we obtain:
\begin{equation*}
    \label{eq.xchain}
X_{\text{opt}} - X_{\text{chain}} > \left (1- \frac{1}{c(2\beta
    \sqrt{2} -1)} \right) X_{\text{opt}} = K\times X_{\text{opt}}\; . 
\end{equation*}

Since $\frac{S}{pD}\leq  \frac{1}{1+\varepsilon^*} \fr$, we have
 $\frac{S}{pD}\leq \left(1-\frac{1}{\sqrt{2} cpK}\right) \fr$, and
$\fr - \frac{S}{pD} \geq \frac{\fr}{\sqrt{2} cpK}$. Since
$X_{\text{opt}} = c(pD\fr -S)$ and $K>0$, we obtain
%
%First note that since $c\approx 0.2838$ and 
%$\beta = \max(2-\frac{3}{2p+1}, 2- \frac{p+2}{4p+2})$,  if $p \geq 2$, then $K>0$.
%Since
%$c\approx 0.2838$ and $\beta=2-\frac{3}{2p+1}$, note that we have
%$K<0$ for $p<4$, and $K>0$ otherwise.  
%There are two cases:
% \begin{enumerate}
%     \item $\frac{1}{K} X_{\text{opt}}\geq \frac{1}{\sqrt{2}}D\fr$;
%     \item $\frac{1}{K} X_{\text{opt}}< \frac{1}{\sqrt{2}}D\fr$.
% \end{enumerate}
%
%\begin{itemize}
%\paragraph{$\bullet$} If 
$K \times X_{\text{opt}}\geq \frac{1}{\sqrt{2}}D\fr$, 
and therefore we have 
%, then applying Equation~\eqref{eq.xchain}, we
%obtain 
$X_{\text{opt}} - X_{\text{chain}} >
\frac{1}{\sqrt{2}}D\fr$. This means that each task that can be
re-executed in any solution to \indep is indeed re-executed in the
solution given by \approxchain, since all these tasks have a weight
lower than $\frac{1}{\sqrt{2}}D\fr$. Since $X_{\text{opt}}$ is greater
than the total weight of the tasks that can be re-executed, we can use
Theorem~\ref{thm.chain} in the case $p=1$, on the subset of
tasks~$T_i$ such that $w_i\leq\frac{1}{\sqrt{2}}D\fr$. The other tasks
are executed once at speed~$\fr$. We define $\finf = \frac{w_i}{1.9
  D}$, so that $\finf < \frac{1}{1.9\sqrt{2}} \fr < \frac{2c}{1+c}
\fr$ and we can apply Theorem~\ref{thm.chain}.  Then, in polynomial
time, we have the optimal solution with new execution speeds:
$\tilde{f_i}^{\text{chain}}$.  Furthermore for each task~$T_i$,
necessarily $$\frac{w_i}{\tilde{f_i}^{\text{chain}}} \leq
\frac{w_i}{\finf} = 1.9 D.$$

Note that since $p\geq 24$, we have $\beta \geq 1.9$, and
$\frac{w_i}{\tilde{f_i}^{\text{chain}}} \leq \beta D$. We can
therefore schedule the new tasks~$\tilde{T_i}$ within the deadline
relaxation using \dff, as a direct consequence of
Lemma~\ref{lemma.bigload}. %, similarly to the case~2 above (schedule
%tasks such that $D\leq \frac{w_i}{\tilde{f_i}^{\text{chain}}} \leq
%\beta D$ alone on a processor, and use Lemma~\ref{lemma.bigload} for
%the remaining tasks).
\end{proof}

We can conclude by stating that thanks to
Propositions~\ref{prop.epsilon1} and~\ref{prop.2}, since
$\varepsilon^*$ is in $\Theta(\frac{1}{p})$ and $\beta$ is in
$2-\Theta(\frac{1}{p})$, this algorithm is a
$(1+\Theta(\frac{1}{p}),2-\Theta(\frac{1}{p}))$-approximation. 
Indeed, $\varepsilon^*<1$ and therefore
$(1+\varepsilon^*)^2<1+3\varepsilon^*$.

Furthermore, the algorithm is polynomial in the size of the input and
in $\frac{1}{\varepsilon^*}$.

\end{document}